\documentclass[aps,prl,twocolumn,10pt]{revtex4-2}

\usepackage{mathrsfs}
\usepackage{bbm}
\usepackage{amsmath}
\usepackage{amssymb}
\usepackage{amsthm}
\usepackage{graphicx}
\usepackage{MnSymbol}
\usepackage{mathtools}
\usepackage[italicdiff]{physics}

\makeatletter
\usepackage{hyperref}
\usepackage{color}
\definecolor{supcol}{RGB}{10,50,180}
\definecolor{eqcol}{RGB}{220,10,100}
\hypersetup{
	colorlinks,
	citecolor=supcol,
	linkcolor=eqcol,
	urlcolor=supcol
}

\allowdisplaybreaks

\newtheorem{theorem}{Theorem}

\newtheorem{definition}[theorem]{Definition}

\newtheorem{result}[theorem]{Result}

\newcommand{\mbb}{\mathbb}

\newcommand{\msf}{\mathsf}

\newcommand{\sectionprl}[1]{{\em #1}\/.---}

\newcommand{\vpt}{\vb*{p}_t}
\newcommand{\vecp}{\vb*{p}}

\newcommand{\vecq}{\vb*{q}}
\newcommand{\vecr}{\vb*{r}}
\newcommand{\vecx}{\vb*{x}}
\newcommand{\vecv}{\vb*{v}}
\newcommand{\vpi}{\vb*{\pi}}

\begin{document}
\title{Thermomajorization Mpemba Effect}

\author{Tan Van Vu}
\email{tan.vu@yukawa.kyoto-u.ac.jp}
\affiliation{Center for Gravitational Physics and Quantum Information, Yukawa Institute for Theoretical Physics, Kyoto University, Kitashirakawa Oiwakecho, Sakyo-ku, Kyoto 606-8502, Japan}

\author{Hisao Hayakawa}
\affiliation{Center for Gravitational Physics and Quantum Information, Yukawa Institute for Theoretical Physics, Kyoto University, Kitashirakawa Oiwakecho, Sakyo-ku, Kyoto 606-8502, Japan}

\date{\today}

\begin{abstract}
The Mpemba effect is a counterintuitive physical phenomenon where a hot system cools faster than a warm one. In recent years, theoretical analyses of the Mpemba effect have been developed for microscopic systems and experimentally verified. However, the conventional theory relies on a specific choice of distance measure to quantify relaxation speed, leading to several theoretical ambiguities. In this Letter, we derive a rigorous quantification of the Mpemba effect based on thermomajorization theory, referred to as the thermomajorization Mpemba effect. This approach resolves all existing ambiguities and provides a unification of the conventional Mpemba effect across all monotone measures. Furthermore, we demonstrate the generality of the thermomajorization Mpemba effect for Markovian dynamics, rigorously proving that it can occur in any temperature regime with fixed energy levels.
\end{abstract}

\maketitle

\sectionprl{Introduction}Thermal relaxation processes, where an out-of-equilibrium system relaxes toward a thermal state, can exhibit counterintuitive phenomena. 
A representative example is the Mpemba effect \cite{Mpemba.1969.PE}, where a system initially in a hot state cools faster than one starting in a warm state. 
This surprising behavior was first observed in macroscopic systems, such as water \cite{Jeng.2006.AJP} and other substances \cite{Greaney.2011.MMT,Ahn.2016.KJCE}, and has been numerically examined in diverse dynamics \cite{Lasanta.2017.PRL,BaityJesi.2019.PNAS,Torrente.2019.PRE,Gijon.2019.PRE,Biswas.2020.PRE,Santos.2020.PF,Yang.2020.PRE,Takada.2021.PRE,Patron.2021.PRE,Megias.2022.PRE,Chatterjee.2024.PRE}. 
In recent years, the theory of the Mpemba effect has been extensively studied for microscopic systems \cite{Lu.2017.PNAS,Nava.2019.PRB,Klich.2019.PRX,Gal.2020.PRL,Walker.2021.JSM,Carollo.2021.PRL,Busiello.2021.NJP,Degnther.2022.EPL,Kochsiek.2022.PRA,Schwarzendahl.2022.PRL,Teza.2023.PRL,Chatterjee.2023.PRL,Ivander.2023.PRE,Walker.2023.arxiv,Chatterjee.2024.PRA,Wang.2024.PRR,Moroder.2024.PRL,Pemartin.2024.PRL,Nava.2024.PRL,Strachan.2024.arxiv}, leading to investigations into its variants, including the inverse Mpemba effect \cite{Lu.2017.PNAS,Lasanta.2017.PRL,Shapira.2024.PRL}, cooling-heating asymmetry \cite{Lapolla.2020.PRL,Vu.2021.PRR,Manikandan.2021.PRR,Meibohm.2021.PRE,Ibez.2024.NP}, and athermal scenarios \cite{Ares.2023.NC,Murciano.2024.JSM,Rylands.2024.PRL,Liu.2024.PRL,Yamashika.2024.PRB,Chang.2024.arxiv}. 
On the experimental front, the Mpemba effect has been confirmed across various platforms, from classical to quantum systems \cite{Kumar.2020.N,Kumar.2022.PNAS,Joshi.2024.PRL}.

In the conventional theory of the Mpemba effect in Markovian dynamics \cite{Lu.2017.PNAS,Klich.2019.PRX}, a specific distance measure is used to quantify relaxation speed and determine which process is faster. 
These distance measures must satisfy certain conditions, such as monotonicity and convexity \cite{Lu.2017.PNAS}. 
However, even with these conditions, there are infinitely many possible measures, and no logical reason to favor one over another. 
Importantly, the Mpemba effect---a physical phenomenon---should not depend on the specific choice of distance measure \cite{Biswas.2023.PRE}. 
This naturally suggests that \emph{all} monotone measures should be considered simultaneously to rigorously quantify relaxation speed and establish a true criterion for the Mpemba effect's occurrence.
A critical issue, however, is that the time to observe the Mpemba effect under certain measures can become arbitrarily large. Therefore, it is essential to determine whether the effect can occur for all monotone measures within a \emph{finite bounded} time---a challenge that is nontrivial due to the infinite number of such measures.

\begin{figure}[b]
\centering
\includegraphics[width=1\linewidth]{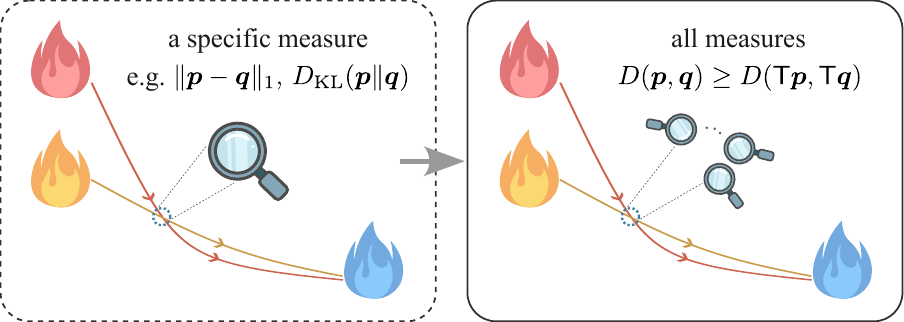}
\protect\caption{The left panel illustrates the conventional theory of the Mpemba effect, where a specific monotone measure, such as the total variation distance and the Kullback-Leibler divergence, is used to quantify the relaxation speed. The right panel demonstrates the proposed thermomajorization Mpemba effect, where all monotone measures are simultaneously exploited to evaluate the relaxation trend.}\label{fig:Cover}
\end{figure}

In this Letter, we address and fully resolve the aforementioned problem. 
To achieve this, we utilize thermomajorization theory \cite{Ruch.1978.JCP,Horodecki.2013.NC}, a mathematical framework widely applied in statistics \cite{Marshall.2011}, quantum information theory \cite{Streltsov.2017.RMP,Chitambar.2019.RMP,Bartosik.2024.RMP}, and quantum thermodynamics \cite{Brando.2015.PNAS,Lostaglio.2015.NC}. 
By considering all monotone measures to evaluate relaxation speed, we derive a rigorous quantification for the occurrence of the Mpemba effect within a finite time, which we refer to as the thermomajorization Mpemba effect (Fig.~\ref{fig:Cover}). 
This yields a unification for the Mpemba effect, as the thermomajorization Mpemba effect implies the occurrence of the Mpemba effect within a bounded time for any monotone measure, and vice versa.
We also establish a criterion for Markovian dynamics to exhibit the thermomajorization Mpemba effect in the long-time regime, refining the conventional criterion that applies to specific measures. 
Although the notion of the thermomajorization Mpemba effect is stricter than the conventional one, we demonstrate its generality by proving that it can occur at arbitrary temperatures with fixed energy levels.
Our results provide a deeper understanding of the Mpemba effect, particularly through the unified quantification and the generality of its occurrence.

\sectionprl{Setup}We consider the thermal relaxation process of an open system with $d\,(\ge 3)$ internal states.
The system is coupled to a thermal reservoir at the inverse temperature $\beta=(k_BT)^{-1}$, where the Boltzmann constant $k_B$ is set to unity for simplicity.
The interactions with the thermal reservoir induce stochastic transitions between states, and the dynamics of the system is governed by the master equation,
\begin{equation}
\dot{\vecp}_t=\msf{W}\vpt,
\end{equation}
where $\vpt=[p_1(t),\dots,p_d(t)]^\top$ represents the probability distribution of the system at time $t$, and the matrix $\msf{W}=[w_{mn}]\in\mbb{R}^{d\times d}$ is time-independent and irreducible, with $w_{mn}\ge 0$ denoting the transition rate from state $n$ to state $m\,(\neq n)$, and $\sum_{m=1}^dw_{mn}=0$.
Without loss of generality, we assume that the energy levels of states satisfy $0\le \epsilon_1\le\dots\le \epsilon_d$, where $\epsilon_n$ denotes the energy level of state $n$.
The transition rates fulfill the detailed balance condition, $w_{mn}e^{-\beta \epsilon_n}=w_{nm}e^{-\beta \epsilon_m}$, which ensures that the system always relaxes toward the Gibbs thermal state $\vpi$, regardless of the initial state.
For $m\neq n$, the transition rate $w_{mn}$ can be expressed in a general form of
\begin{equation}
w_{mn}=e^{-\beta(b_{mn}-\epsilon_n)},
\end{equation}
where $b_{mn}=b_{nm}$ are the barrier coefficients.
In the conventional setup of the Mpemba effect, the relaxation of two different thermal states  is typically considered: a hot state $\vpi^h$ at temperature $T_h$ and a warm state $\vpi^w$ at temperature $T_w$, with $T<T_w<T_h$ \cite{fnt1}.
For notational convenience, we use the superscripts $h$ and $w$ to denote the corresponding quantities when the system is initialized in the hot and warm thermal states, respectively.

Let us explain the conventional theory of the Mpemba effect and its ambiguities.
The focus is on the relaxation speed of the hot state and the warm state toward equilibrium.
To examine this, two relevant time regimes are typically considered in the literature: the intermediate- and long-time regions.

In the intermediate-time regime, a distance measure $D$ is commonly used to quantify how close the current state is to the final thermal state $\vpi$ \cite{Lu.2017.PNAS,Klich.2019.PRX}. 
Initially, it is required that $D(\vpi^h,\vpi)>D(\vpi^w,\vpi)$ should be satisfied, meaning the hot state is farther from the final state than the warm state.
The Mpemba effect is said to occur at time $t$ if a crossover between the relaxation processes happens, that is,
\begin{equation}
	D(\vpt^h,\vpi)<D(\vpt^w,\vpi).
\end{equation}
While this quantification is intuitive and reasonable to some extent, an ambiguity remains in the choice of the distance measure $D$, as it is not unique.
Commonly used monotone measures include the total variation distance and the Kullback-Leibler (KL) divergence. 
However, the physical phenomenon should ideally be independent of the specific choice of measure. 
Notably, the crossover time for some measures may be unbounded, i.e., $\sup_D t_D = \infty$, where $t_D$ denotes the crossover time for a given measure $D$.

In the long-time regime, another way to quantify the relaxation speed is by examining the slowest decay mode.
Let $0=\lambda_1>\lambda_2\ge\dots\ge\lambda_d$ represent the eigenvalues of the transition matrix and $\{\vecr_n\}$ denote the corresponding right eigenvectors, i.e., $\msf{W}\vecr_n=\lambda_n\vecr_n$.
Both the eigenvalues and eigenvectors are real numbers because the matrix $\msf{W}$ satisfies the detailed balance condition \cite{Schnakenberg.1976.RMP}.
Since $\{\vecr_n\}$ form a basis for the space $\mbb{R}^d$, the initial distribution $\vpi^s$ can be expressed as a linear combination of $\{\vecr_n\}$ as $\vpi^s=\vpi+\sum_{n=2}^{d}a_n^s\vecr_n$, where $\{a_n^s\}$ are real numbers and $s\in\{h,w\}$.
Consequently, the probability distribution at time $t$ can be analytically calculated as
\begin{equation}
\vpt^s=\vpi+\sum_{n=2}^{d}a_n^se^{\lambda_n t}\vecr_n.\label{eq:time.eigvec.decomp}
\end{equation}
In the nondegenerate case (i.e., $\lambda_2>\lambda_3$), the probability distribution $\vpt^s$ can be approximated in the long-time regime by the slowest decay mode as $\vpt^s=\vpi+a_2^se^{\lambda_2 t}\vecr_2+O(e^{\lambda_3t})$.
Thus, the relaxation speed may be quantified by the value of $|a_2^s|$ \cite{Lu.2017.PNAS}, which can be obtained analytically in closed form \cite{Klich.2019.PRX}.
Conventionally, the Mpemba effect is said to occur if $|a_2^h|<|a_2^w|$.
However, this condition does not ensure that $\sup_Dt_D<\infty$. 
In other words, there may exist monotone measures where the Mpemba effect fails to manifest within a \emph{finite} time under this condition.
Additionally, this approach breaks down in the presence of degeneracy (i.e., when $\lambda_2=\lambda_3$).
In this Letter, we develop a rigorous framework for the Mpemba effect that addresses these ambiguities and provides a unification for the conventional approach.

\sectionprl{Main results}As mentioned earlier, the existence of the Mpemba effect should not depend on the specific choice of a distance measure. 
Instead, all monotone measures must be considered to determine whether the relaxation trend changes. 
Given the infinite number of monotone measures, this task might seem infeasible. 
Nevertheless, we can fully resolve this issue using the thermomajorization theory.

Before explaining our main results, we first briefly describe the concept of thermomajorization \cite{Horodecki.2013.NC}, a generalization of classical majorization theory \cite{Marshall.2011,Sagawa.2022}. 
Roughly speaking, thermomajorization compares the spread of probability distributions relative to a reference Gibbs distribution.
We introduce thermomajorization using the notion of the generalized Lorenz curve.
For any pair of probability distributions $(\vecp,\vecq)$, let $\tilde{\vecp}$ and $\tilde{\vecq}$ be their respective permutations, where their components are rearranged in the same way such that $\tilde{p}_1/\tilde{q}_1\ge \tilde{p}_2/\tilde{q}_2\ge\dots\ge \tilde{p}_d/\tilde{q}_d$ holds.
The Lorenz curve of this pair is defined as a concave polyline that connects the origin $(0,0)$ and the points $(\sum_{k=1}^n\tilde{q}_k,\sum_{k=1}^n\tilde{p}_k)$ for $n=1,\dots,d$.
We say that $\vecp$ thermomajorizes $\vecp'$ with respect to $\vpi$, denoted by $\vecp'\prec_{\vpi}\vecp$, if the Lorenz curve of $(\vecp,\vpi)$ lies above that of $(\vecp',\vpi)$ (see Fig.~\ref{fig:Majorization} for illustration).
Since $\vpi\prec_{\vpi}\vecp$ holds for any distribution $\vecp$, the thermomajorization order $\prec_{\vpi}$ can be viewed as a measure of how ``close'' the distribution $\vecp$ is to $\vpi$ \cite{Sagawa.2022}.
A remarkable property of thermomajorization is that $\vecp'\prec_{\vpi}\vecp$ is equivalent to any of the following conditions \cite{Blackwell.1953.AMS}:
\begin{enumerate}
	\item[(a)] There exists a stochastic matrix $\msf{T}$ such that $\msf{T}\vecp=\vecp'$ and $\msf{T}\vpi=\vpi$.
	\item[(b)] $\sum_{n=1}^d\pi_nf(p_n'/\pi_n)\le\sum_{n=1}^d\pi_nf(p_n/\pi_n)$ for any continuous, convex function $f(x)$.
	\item[(c)] $\|\vecp'-z\vpi\|_1\le\|\vecp-z\vpi\|_1~\forall z\in\mbb{R}$, where $\|\vecx\|_1\coloneqq\sum_{n}|x_n|$.
\end{enumerate}
Here, the stochastic matrix $\msf{T}$ represents a Markov chain with nonnegative elements and satisfies  $\vb*{1}^\top\msf{T}=\vb*{1}^\top$, where $\vb*{1}$ is the all-one vector.
In addition to the graphical representation of the Lorenz curves, thermomajorization can also be efficiently verified by checking condition (c) for $z\in\{p_1/\pi_1,\dots,p_d/\pi_d\}$ \cite{vomEnde.2022.LAA}.
With this foundational theory, we now present our main results.

\begin{figure}[t]
\centering
\includegraphics[width=1\linewidth]{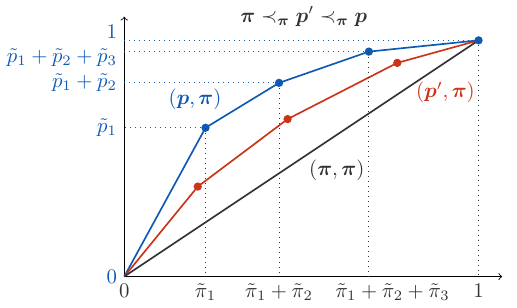}
\protect\caption{Graphical illustration of thermomajorization. The Lorenz curves of the pairs $(\vecp,\vpi)$, $(\vecp',\vpi)$, and $(\vpi,\vpi)$ are geometrically ordered, implying $\vpi\prec_{\vpi}\vecp'\prec_{\vpi}\vecp$.}\label{fig:Majorization}
\end{figure}

We begin by introducing the notion of the thermomajorization Mpemba effect, which can be regarded as a unification of the conventional Mpemba effect based on specific distance measures. 
It is important to note that $\vpi^w\prec_{\vpi}\vpi^h$ (see Sec.~S1 in Ref.~\cite{Supp.PhysRev} for the proof), meaning the warm state is closer to the final thermal state than the hot state from the perspective of thermomajorization.
The essence of the Mpemba effect---namely, the change in the relaxation trend---leads us to the following definition of the thermomajorization Mpemba effect.
\begin{definition}\label{def:TME}
The thermomajorization Mpemba effect is said to occur at time $t$ if $\vecp_t^w$ thermomajorizes $\vecp_t^h$ with respect to $\vpi$, i.e., the following relation holds true:
\begin{equation}\label{eq:thermo.Mpemba.eff}
	\vecp_t^h\prec_{\vpi}\vecp_t^w.
\end{equation}
\end{definition}
Relation \eqref{eq:thermo.Mpemba.eff} indicates that, at time $t$, the relaxation from the hot state becomes faster than that from the warm state from the viewpoint of thermomajorization. 
Therefore, the thermomajorization Mpemba effect shares the same essence as the conventional Mpemba effect: both exhibit a crossover in the relaxation processes of the hot and warm states.
The key difference lies in how the ``distance'' between the current state and the final thermal state is quantified. 
In our definition, the distance is measured by thermomajorization, which is independent of specific distance measures, whereas the conventional definition relies on a particular measure. 
Notably, the thermomajorization Mpemba effect implies the conventional one based on all $f$-divergences, according to condition (b).
They include the total variation distance and the KL divergence, with $f(x)=|x-1|$ and $f(x)=x\ln x$, respectively.
Most importantly, the thermomajorization Mpemba effect can be regarded as a unification of conventional quantifications, as demonstrated in the following result.
The proof is presented in Sec.~S2 of Ref.~\cite{Supp.PhysRev}.
\begin{result}\label{res:equiv.TME}
The thermomajorization Mpemba effect is equivalent to the Mpemba effect occurring within a finite bounded time for all monotone measures.
\end{result}
This result indicates that the thermomajorization Mpemba effect guarantees bounded crossover times for all monotone measures. 
Unlike the conventional approach, which depends on a specific measure, our framework incorporates \emph{all} monotone measures to rigorously determine relaxation speed. 
Relying on a single measure is insufficient to capture changes in the relaxation trend, as other measures may fail to exhibit a crossover within a bounded time.
Our new definition therefore offers a rigorous quantification and unification of the Mpemba effect.
Notably, the theory developed thus far (Definition \ref{def:TME} and Result \ref{res:equiv.TME}) is general and dynamics-independent \cite{fnt2}, allowing it to be applied to other phenomena, such as the inverse Mpemba effect, as well as to different dynamics, including active systems \cite{Datta.2022.PRX,Biswas.2024.arxiv}. 
Additionally, it can be extended to classical continuous-variable and quantum systems, as the extension of thermomajorization theory to these cases is straightforward \cite{Sagawa.2022,Supp.PhysRev}.

\begin{figure*}[t]
\centering
\includegraphics[width=1\linewidth]{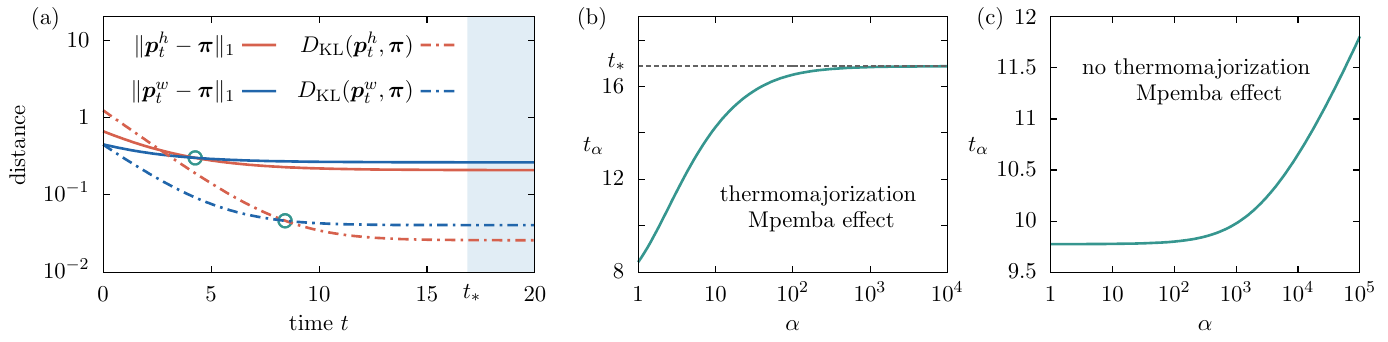}
\protect\caption{Numerical illustration of the thermomajorization Mpemba effect in the three-level system. (a) The crossover occurs at times $t_{\rm TV}\approx 4.26$ and $t_{\rm KL}\approx 8.43$ when measured by the total variation distance and the KL divergence, respectively. In contrast, the thermomajorization Mpemba effect occurs at time $t_*\approx 16.87$. (b) The crossover time $t_\alpha$, measured by the $f_\alpha$-divergence, increases and saturates at $t_*$ as $\alpha\to\infty$. (c) The divergence of the crossover time $t_\alpha$ indicates the absence of the thermomajorization Mpemba effect, even though $|a_2^h|<|a_2^w|$.}\label{fig:ThreeLevel}
\end{figure*}

So far, we have discussed a general theory for quantifying the Mpemba effect across the entire time regime. 
It is natural to investigate the existence of the thermomajorization Mpemba effect in the \emph{long-time} regime.
We note that the results presented below are restricted to Markovian dynamics.
Consider the generic case where $0=\lambda_1>\lambda_2=\dots=\lambda_\kappa>\lambda_{\kappa+1}$, with a degeneracy degree of $\kappa-1$ for the second largest eigenvalue.
In this case, the distribution at long time $t$ can be approximated as $\vpt^s=\vpi+e^{\lambda_2t}\vecv^s+O(e^{\lambda_{\kappa+1}t})$, where $\vecv^s\coloneqq\sum_{n=2}^\kappa a_n^s\vecr_n$.
The occurrence of the thermomajorization Mpemba effect in the regime of $t\gg 1$ can be determined through $\vecv^w$ and $\vecv^h$, which encode information about the degenerate eigenvectors. 
According to condition (c), we obtain the following criterion (see Sec.~S3 of Ref.~\cite{Supp.PhysRev} for the proof).
\begin{result}\label{res:long.time.tME.cond}
The thermomajorization Mpemba effect occurs in the long-time regime if and only if the following inequalities hold:
\begin{equation}\label{eq:tME.longtime.cond}
	\|\vecv^h-z\vpi\|_1\le\|\vecv^w-z\vpi\|_1~\forall z\in\qty{\frac{v_k^w}{\pi_k},\frac{v_k^h}{\pi_k}}_{1\le k\le d},
\end{equation}
with at least one of them holding strictly.
\end{result}
Inequality \eqref{eq:tME.longtime.cond} provides a simple criterion for determining the thermomajorization Mpemba effect. 
In the nondegenerate case (i.e., $\lambda_2>\lambda_3$), this criterion can be simplified to
\begin{equation}\label{eq:nondeg.cond}
	\Big\|\frac{a_2^h}{a_2^w}\vecr_2-z\vpi\Big\|_1\le\|\vecr_2-z\vpi\|_1\le\Big\|\frac{a_2^w}{a_2^h}\vecr_2-z\vpi\Big\|_1,
\end{equation}
where $z\in\{r_{2k}/\pi_k\}_{1\le k\le d}$.
Using this criterion, we derive the following result, which demonstrates the connection between our approach and the conventional one in the long-time regime.
The proof is given in Sec.~S4 of Ref.~\cite{Supp.PhysRev}.
\begin{result}\label{res:eig.TME}
For any transition matrix $\msf{W}$ with $\lambda_2>\lambda_3$, the condition $|a_2^h|<|a_2^w|$ is necessary but not sufficient to guarantee the occurrence of the thermomajorization Mpemba effect in the long-time regime.
\end{result}
This result shows that while $|a_2^h|<|a_2^w|$ is a necessary condition, it is insufficient to achieve the thermomajorization Mpemba effect.
Technically speaking, $|a_2^h|<|a_2^w|$ does not imply inequality \eqref{eq:nondeg.cond}. 
If inequality \eqref{eq:nondeg.cond} is violated, one can always find a monotone measure (e.g., a $f$-divergence) for which the Mpemba effect is not observed within a finite time.

Lastly, we discuss the generality of the thermomajorization Mpemba effect, particularly regarding its occurrence across different temperature regimes. 
Since the definition of the thermomajorization Mpemba effect is stricter than the conventional one, a nontrivial and relevant question arises: in which temperature regimes can the Mpemba effect occur? 
From both theoretical and practical perspectives, it is important to determine whether the Mpemba effect can emerge at extremely low or high temperatures. 
We affirmatively answer this question in the following result, with the proof provided in Sec.~S5 of Ref.~\cite{Supp.PhysRev}.
\begin{result}\label{res:uni.TME}
For arbitrary temperatures $T<T_w<T_h$ and fixed energy levels $\{\epsilon_n\}$, there always exists a transition matrix that showcases the thermomajorization Mpemba effect.
\end{result}
This result demonstrates the generality of the thermomajorization Mpemba effect across the entire temperature range. 
For arbitrary fixed temperatures and energy levels, it is always possible to find an appropriate set of barrier coefficients that will lead to the thermomajorization Mpemba effect. 
This provides a valuable guideline for designing systems that can exhibit the Mpemba effect at any temperature.

\sectionprl{Numerical demonstration}We illustrate the thermomajorization Mpemba effect using a minimal three-level system \cite{Lu.2017.PNAS}.
The energy levels are $\vb*{\epsilon}=[0,0.1,0.7]^\top$, and the temperatures for the hot, warm, and final states are fixed at $T_h=1.3$, $T_w=0.42$, and $T=0.1$, respectively.
The barrier coefficients are set as $b_{12}=1.5$, $b_{13}=0.8$, and $b_{23}=1.2$.
It can be numerically confirmed that the crossover occurs at time $t_{\rm TV}\approx 4.26$ and $t_{\rm KL}\approx 8.43$ when the total variation distance and the KL divergence are used to measure the distance to the final thermal state, respectively [see Fig.~\ref{fig:ThreeLevel}(a)]. 
It is also clear that the crossover time strongly depends on the choice of measure.
Interestingly, we find that the thermomajorization Mpemba effect occurs at time $t_*\approx 16.87$ \cite{fnt3}, which is relatively late compared to the conventional approach.
This means that before time $t_*$, there always exists a monotone measure $D$ (e.g., a $f$-divergence) such that the relaxation from the hot state remains slower than that from the warm state when evaluated by measure $D$.
Conversely, after time $t_*$, the relaxation from the hot state becomes faster than that from the warm state when assessed by \emph{any} monotone measure. 

Next, we provide a detailed analysis of the underlying origin of the crossover time $t_*$.
To this end, consider a class of $f$-divergence, defined for $\alpha\ge 1$ as follows:
\begin{equation}
	f_\alpha(x)=\begin{cases}
		(x^\alpha-\alpha x+\alpha-1)/\alpha(\alpha-1) & \text{if}~\alpha>1,\\
		x\ln x-x+1 & \text{if}~\alpha=1.
	\end{cases}
\end{equation}
For each $\alpha\ge 1$, we use the $f_\alpha$-divergence to measure the distance to equilibrium.
The time $t_\alpha$ at which the crossover occurs is plotted in Fig.~\ref{fig:ThreeLevel}(b).
As shown, $t_\alpha$ converges to $t_*$ as $\alpha$ increases.
Therefore, for any time $t<t_*$, there always exists a $f_\alpha$-divergence such that the crossover has not yet occurred at time $t$.
In this case, the crossover time $t_*$ originates from the $f_\infty$-divergence, which evaluates the maximum ratio of probabilities $\max_n[p_n(t)/\pi_n]$ (see Sec.~S6 in Ref.~\cite{Supp.PhysRev} for details).

Lastly, we demonstrate a case where the thermomajorization Mpemba effect does not occur, even though $|a_2^h|<|a_2^w|$.
The barrier coefficients are appropriately adjusted so that $|a_2^h|<|a_2^w|$, but inequality \eqref{eq:nondeg.cond} is violated \cite{fnt4}, while the energy levels and temperatures remain as previously specified.
We analogously calculate the crossover time for each $f_\alpha$-divergence and plot the results in Fig.~\ref{fig:ThreeLevel}(c).
As shown, $t_\alpha$ does not converge and tends to diverge as $\alpha$ increases.
In other words, for any time $t$, one can always find a $f_\alpha$-divergence such that the Mpemba effect is not observed with this measure.

\sectionprl{Conclusion and discussion}Using thermomajorization theory, we developed a rigorous quantification and unification of the Mpemba effect. 
This unified framework resolves the ambiguities inherent in the conventional approach, which relies on specific distance measures. 
Furthermore, we demonstrated the generality of the thermomajorization Mpemba effect across arbitrary temperatures. 
Our study not only advances the understanding of the Mpemba effect but also establishes a foundational framework for exploring its variants.
While we presented our theory for discrete-state systems, it can also be straightforwardly extended to continuous-state cases (see Sec.~S7 in Ref.~\cite{Supp.PhysRev}).

In addition to quantifications based on probability distributions, some studies employed energetic observables to quantify the Mpemba effect \cite{Biswas.2023.PRE, Ohga.2024.arxiv}. 
Investigating the connection between these approaches would be an interesting direction for future research. 
Another important question is determining the minimum timescale at which the thermomajorization Mpemba effect occurs, which could potentially be addressed through the lens of speed limits \cite{Deffner.2017.JPA,Vu.2023.PRL.TSL,Srivastav.2024.arxiv}.

\begin{acknowledgments}
Fruitful discussions with Fr{\'e}d{\'e}ric van Wijland and Keiji Saito are gratefully acknowledged.
{T.V.V} thanks Naruo Ohga and Sosuke Ito for sharing their preliminary manuscript with H.~Hayakawa.
This work was supported by JSPS KAKENHI Grant No.~JP23K13032 and the Kyoto University Foundation.
\end{acknowledgments}


\begin{thebibliography}{74}%
\makeatletter
\providecommand \@ifxundefined [1]{%
 \@ifx{#1\undefined}
}%
\providecommand \@ifnum [1]{%
 \ifnum #1\expandafter \@firstoftwo
 \else \expandafter \@secondoftwo
 \fi
}%
\providecommand \@ifx [1]{%
 \ifx #1\expandafter \@firstoftwo
 \else \expandafter \@secondoftwo
 \fi
}%
\providecommand \natexlab [1]{#1}%
\providecommand \enquote  [1]{``#1''}%
\providecommand \bibnamefont  [1]{#1}%
\providecommand \bibfnamefont [1]{#1}%
\providecommand \citenamefont [1]{#1}%
\providecommand \href@noop [0]{\@secondoftwo}%
\providecommand \href [0]{\begingroup \@sanitize@url \@href}%
\providecommand \@href[1]{\@@startlink{#1}\@@href}%
\providecommand \@@href[1]{\endgroup#1\@@endlink}%
\providecommand \@sanitize@url [0]{\catcode `\\12\catcode `\$12\catcode
  `\&12\catcode `\#12\catcode `\^12\catcode `\_12\catcode `\%12\relax}%
\providecommand \@@startlink[1]{}%
\providecommand \@@endlink[0]{}%
\providecommand \url  [0]{\begingroup\@sanitize@url \@url }%
\providecommand \@url [1]{\endgroup\@href {#1}{\urlprefix }}%
\providecommand \urlprefix  [0]{URL }%
\providecommand \Eprint [0]{\href }%
\providecommand \doibase [0]{https://doi.org/}%
\providecommand \selectlanguage [0]{\@gobble}%
\providecommand \bibinfo  [0]{\@secondoftwo}%
\providecommand \bibfield  [0]{\@secondoftwo}%
\providecommand \translation [1]{[#1]}%
\providecommand \BibitemOpen [0]{}%
\providecommand \bibitemStop [0]{}%
\providecommand \bibitemNoStop [0]{.\EOS\space}%
\providecommand \EOS [0]{\spacefactor3000\relax}%
\providecommand \BibitemShut  [1]{\csname bibitem#1\endcsname}%
\let\auto@bib@innerbib\@empty
%</preamble>
\bibitem [{\citenamefont {Mpemba}\ and\ \citenamefont
  {Osborne}(1969)}]{Mpemba.1969.PE}%
  \BibitemOpen
  \bibfield  {author} {\bibinfo {author} {\bibfnamefont {E.~B.}\ \bibnamefont
  {Mpemba}}\ and\ \bibinfo {author} {\bibfnamefont {D.~G.}\ \bibnamefont
  {Osborne}},\ }\bibfield  {title} {\bibinfo {title} {{Cool?}},\ }\href
  {https://doi.org/10.1088/0031-9120/4/3/312} {\bibfield  {journal} {\bibinfo
  {journal} {Phys. Educ.}\ }\textbf {\bibinfo {volume} {4}},\ \bibinfo {pages}
  {172} (\bibinfo {year} {1969})}\BibitemShut {NoStop}%
\bibitem [{\citenamefont {Jeng}(2006)}]{Jeng.2006.AJP}%
  \BibitemOpen
  \bibfield  {author} {\bibinfo {author} {\bibfnamefont {M.}~\bibnamefont
  {Jeng}},\ }\bibfield  {title} {\bibinfo {title} {{The Mpemba effect: When can
  hot water freeze faster than cold?}},\ }\href
  {https://doi.org/10.1119/1.2186331} {\bibfield  {journal} {\bibinfo
  {journal} {Am. J. Phys.}\ }\textbf {\bibinfo {volume} {74}},\ \bibinfo
  {pages} {514} (\bibinfo {year} {2006})}\BibitemShut {NoStop}%
\bibitem [{\citenamefont {Greaney}\ \emph {et~al.}(2011)\citenamefont
  {Greaney}, \citenamefont {Lani}, \citenamefont {Cicero},\ and\ \citenamefont
  {Grossman}}]{Greaney.2011.MMT}%
  \BibitemOpen
  \bibfield  {author} {\bibinfo {author} {\bibfnamefont {P.~A.}\ \bibnamefont
  {Greaney}}, \bibinfo {author} {\bibfnamefont {G.}~\bibnamefont {Lani}},
  \bibinfo {author} {\bibfnamefont {G.}~\bibnamefont {Cicero}},\ and\ \bibinfo
  {author} {\bibfnamefont {J.~C.}\ \bibnamefont {Grossman}},\ }\bibfield
  {title} {\bibinfo {title} {{Mpemba-like behavior in carbon nanotube
  resonators}},\ }\href {https://doi.org/10.1007/s11661-011-0843-4} {\bibfield
  {journal} {\bibinfo  {journal} {Metall. Mater. Trans. A}\ }\textbf {\bibinfo
  {volume} {42}},\ \bibinfo {pages} {3907} (\bibinfo {year}
  {2011})}\BibitemShut {NoStop}%
\bibitem [{\citenamefont {Ahn}\ \emph {et~al.}(2016)\citenamefont {Ahn},
  \citenamefont {Kang}, \citenamefont {Koh},\ and\ \citenamefont
  {Lee}}]{Ahn.2016.KJCE}%
  \BibitemOpen
  \bibfield  {author} {\bibinfo {author} {\bibfnamefont {Y.-H.}\ \bibnamefont
  {Ahn}}, \bibinfo {author} {\bibfnamefont {H.}~\bibnamefont {Kang}}, \bibinfo
  {author} {\bibfnamefont {D.-Y.}\ \bibnamefont {Koh}},\ and\ \bibinfo {author}
  {\bibfnamefont {H.}~\bibnamefont {Lee}},\ }\bibfield  {title} {\bibinfo
  {title} {{Experimental verifications of Mpemba-like behaviors of clathrate
  hydrates}},\ }\href {https://doi.org/10.1007/s11814-016-0029-2} {\bibfield
  {journal} {\bibinfo  {journal} {Korean J. Chem. Eng.}\ }\textbf {\bibinfo
  {volume} {33}},\ \bibinfo {pages} {1903} (\bibinfo {year}
  {2016})}\BibitemShut {NoStop}%
\bibitem [{\citenamefont {Lasanta}\ \emph {et~al.}(2017)\citenamefont
  {Lasanta}, \citenamefont {Vega~Reyes}, \citenamefont {Prados},\ and\
  \citenamefont {Santos}}]{Lasanta.2017.PRL}%
  \BibitemOpen
  \bibfield  {author} {\bibinfo {author} {\bibfnamefont {A.}~\bibnamefont
  {Lasanta}}, \bibinfo {author} {\bibfnamefont {F.}~\bibnamefont {Vega~Reyes}},
  \bibinfo {author} {\bibfnamefont {A.}~\bibnamefont {Prados}},\ and\ \bibinfo
  {author} {\bibfnamefont {A.}~\bibnamefont {Santos}},\ }\bibfield  {title}
  {\bibinfo {title} {{When the hotter cools more quickly: Mpemba effect in
  granular fluids}},\ }\href {https://doi.org/10.1103/PhysRevLett.119.148001}
  {\bibfield  {journal} {\bibinfo  {journal} {Phys. Rev. Lett.}\ }\textbf
  {\bibinfo {volume} {119}},\ \bibinfo {pages} {148001} (\bibinfo {year}
  {2017})}\BibitemShut {NoStop}%
\bibitem [{\citenamefont {Baity-Jesi}\ \emph {et~al.}(2019)\citenamefont
  {Baity-Jesi}, \citenamefont {Calore}, \citenamefont {Cruz}, \citenamefont
  {Fernandez}, \citenamefont {Gil-Narvi{\'{o}}n}, \citenamefont
  {Gordillo-Guerrero}, \citenamefont {I{\~{n}}iguez}, \citenamefont {Lasanta},
  \citenamefont {Maiorano}, \citenamefont {Marinari}, \citenamefont
  {Martin-Mayor}, \citenamefont {Moreno-Gordo}, \citenamefont {Sudupe},
  \citenamefont {Navarro}, \citenamefont {Parisi}, \citenamefont
  {Perez-Gaviro}, \citenamefont {Ricci-Tersenghi}, \citenamefont
  {Ruiz-Lorenzo}, \citenamefont {Schifano}, \citenamefont {Seoane},
  \citenamefont {Taranc{\'{o}}n}, \citenamefont {Tripiccione},\ and\
  \citenamefont {Yllanes}}]{BaityJesi.2019.PNAS}%
  \BibitemOpen
  \bibfield  {author} {\bibinfo {author} {\bibfnamefont {M.}~\bibnamefont
  {Baity-Jesi}}, \bibinfo {author} {\bibfnamefont {E.}~\bibnamefont {Calore}},
  \bibinfo {author} {\bibfnamefont {A.}~\bibnamefont {Cruz}}, \bibinfo {author}
  {\bibfnamefont {L.~A.}\ \bibnamefont {Fernandez}}, \bibinfo {author}
  {\bibfnamefont {J.~M.}\ \bibnamefont {Gil-Narvi{\'{o}}n}}, \bibinfo {author}
  {\bibfnamefont {A.}~\bibnamefont {Gordillo-Guerrero}}, \bibinfo {author}
  {\bibfnamefont {D.}~\bibnamefont {I{\~{n}}iguez}}, \bibinfo {author}
  {\bibfnamefont {A.}~\bibnamefont {Lasanta}}, \bibinfo {author} {\bibfnamefont
  {A.}~\bibnamefont {Maiorano}}, \bibinfo {author} {\bibfnamefont
  {E.}~\bibnamefont {Marinari}}, \bibinfo {author} {\bibfnamefont
  {V.}~\bibnamefont {Martin-Mayor}}, \bibinfo {author} {\bibfnamefont
  {J.}~\bibnamefont {Moreno-Gordo}}, \bibinfo {author} {\bibfnamefont {A.~M.}\
  \bibnamefont {Sudupe}}, \bibinfo {author} {\bibfnamefont {D.}~\bibnamefont
  {Navarro}}, \bibinfo {author} {\bibfnamefont {G.}~\bibnamefont {Parisi}},
  \bibinfo {author} {\bibfnamefont {S.}~\bibnamefont {Perez-Gaviro}}, \bibinfo
  {author} {\bibfnamefont {F.}~\bibnamefont {Ricci-Tersenghi}}, \bibinfo
  {author} {\bibfnamefont {J.~J.}\ \bibnamefont {Ruiz-Lorenzo}}, \bibinfo
  {author} {\bibfnamefont {S.~F.}\ \bibnamefont {Schifano}}, \bibinfo {author}
  {\bibfnamefont {B.}~\bibnamefont {Seoane}}, \bibinfo {author} {\bibfnamefont
  {A.}~\bibnamefont {Taranc{\'{o}}n}}, \bibinfo {author} {\bibfnamefont
  {R.}~\bibnamefont {Tripiccione}},\ and\ \bibinfo {author} {\bibfnamefont
  {D.}~\bibnamefont {Yllanes}},\ }\bibfield  {title} {\bibinfo {title} {{The
  Mpemba effect in spin glasses is a persistent memory effect}},\ }\href
  {https://doi.org/10.1073/pnas.1819803116} {\bibfield  {journal} {\bibinfo
  {journal} {Proc. Natl. Acad. Sci. U.S.A.}\ }\textbf {\bibinfo {volume}
  {116}},\ \bibinfo {pages} {15350} (\bibinfo {year} {2019})}\BibitemShut
  {NoStop}%
\bibitem [{\citenamefont {Torrente}\ \emph {et~al.}(2019)\citenamefont
  {Torrente}, \citenamefont {L\'opez-Casta\~no}, \citenamefont {Lasanta},
  \citenamefont {Reyes}, \citenamefont {Prados},\ and\ \citenamefont
  {Santos}}]{Torrente.2019.PRE}%
  \BibitemOpen
  \bibfield  {author} {\bibinfo {author} {\bibfnamefont {A.}~\bibnamefont
  {Torrente}}, \bibinfo {author} {\bibfnamefont {M.~A.}\ \bibnamefont
  {L\'opez-Casta\~no}}, \bibinfo {author} {\bibfnamefont {A.}~\bibnamefont
  {Lasanta}}, \bibinfo {author} {\bibfnamefont {F.~V.}\ \bibnamefont {Reyes}},
  \bibinfo {author} {\bibfnamefont {A.}~\bibnamefont {Prados}},\ and\ \bibinfo
  {author} {\bibfnamefont {A.}~\bibnamefont {Santos}},\ }\bibfield  {title}
  {\bibinfo {title} {{Large Mpemba-like effect in a gas of inelastic rough hard
  spheres}},\ }\href {https://doi.org/10.1103/PhysRevE.99.060901} {\bibfield
  {journal} {\bibinfo  {journal} {Phys. Rev. E}\ }\textbf {\bibinfo {volume}
  {99}},\ \bibinfo {pages} {060901} (\bibinfo {year} {2019})}\BibitemShut
  {NoStop}%
\bibitem [{\citenamefont {Gij\'on}\ \emph {et~al.}(2019)\citenamefont
  {Gij\'on}, \citenamefont {Lasanta},\ and\ \citenamefont
  {Hern\'andez}}]{Gijon.2019.PRE}%
  \BibitemOpen
  \bibfield  {author} {\bibinfo {author} {\bibfnamefont {A.}~\bibnamefont
  {Gij\'on}}, \bibinfo {author} {\bibfnamefont {A.}~\bibnamefont {Lasanta}},\
  and\ \bibinfo {author} {\bibfnamefont {E.~R.}\ \bibnamefont {Hern\'andez}},\
  }\bibfield  {title} {\bibinfo {title} {{Paths towards equilibrium in
  molecular systems: The case of water}},\ }\href
  {https://doi.org/10.1103/PhysRevE.100.032103} {\bibfield  {journal} {\bibinfo
   {journal} {Phys. Rev. E}\ }\textbf {\bibinfo {volume} {100}},\ \bibinfo
  {pages} {032103} (\bibinfo {year} {2019})}\BibitemShut {NoStop}%
\bibitem [{\citenamefont {Biswas}\ \emph {et~al.}(2020)\citenamefont {Biswas},
  \citenamefont {Prasad}, \citenamefont {Raz},\ and\ \citenamefont
  {Rajesh}}]{Biswas.2020.PRE}%
  \BibitemOpen
  \bibfield  {author} {\bibinfo {author} {\bibfnamefont {A.}~\bibnamefont
  {Biswas}}, \bibinfo {author} {\bibfnamefont {V.~V.}\ \bibnamefont {Prasad}},
  \bibinfo {author} {\bibfnamefont {O.}~\bibnamefont {Raz}},\ and\ \bibinfo
  {author} {\bibfnamefont {R.}~\bibnamefont {Rajesh}},\ }\bibfield  {title}
  {\bibinfo {title} {{Mpemba effect in driven granular Maxwell gases}},\ }\href
  {https://doi.org/10.1103/PhysRevE.102.012906} {\bibfield  {journal} {\bibinfo
   {journal} {Phys. Rev. E}\ }\textbf {\bibinfo {volume} {102}},\ \bibinfo
  {pages} {012906} (\bibinfo {year} {2020})}\BibitemShut {NoStop}%
\bibitem [{\citenamefont {Santos}\ and\ \citenamefont
  {Prados}(2020)}]{Santos.2020.PF}%
  \BibitemOpen
  \bibfield  {author} {\bibinfo {author} {\bibfnamefont {A.}~\bibnamefont
  {Santos}}\ and\ \bibinfo {author} {\bibfnamefont {A.}~\bibnamefont
  {Prados}},\ }\bibfield  {title} {\bibinfo {title} {{Mpemba effect in
  molecular gases under nonlinear drag}},\ }\href
  {https://doi.org/10.1063/5.0016243} {\bibfield  {journal} {\bibinfo
  {journal} {Phys. Fluids}\ }\textbf {\bibinfo {volume} {32}} (\bibinfo {year}
  {2020})}\BibitemShut {NoStop}%
\bibitem [{\citenamefont {Yang}\ and\ \citenamefont
  {Hou}(2020)}]{Yang.2020.PRE}%
  \BibitemOpen
  \bibfield  {author} {\bibinfo {author} {\bibfnamefont {Z.-Y.}\ \bibnamefont
  {Yang}}\ and\ \bibinfo {author} {\bibfnamefont {J.-X.}\ \bibnamefont {Hou}},\
  }\bibfield  {title} {\bibinfo {title} {{Non-Markovian Mpemba effect in
  mean-field systems}},\ }\href {https://doi.org/10.1103/PhysRevE.101.052106}
  {\bibfield  {journal} {\bibinfo  {journal} {Phys. Rev. E}\ }\textbf {\bibinfo
  {volume} {101}},\ \bibinfo {pages} {052106} (\bibinfo {year}
  {2020})}\BibitemShut {NoStop}%
\bibitem [{\citenamefont {Takada}\ \emph {et~al.}(2021)\citenamefont {Takada},
  \citenamefont {Hayakawa},\ and\ \citenamefont {Santos}}]{Takada.2021.PRE}%
  \BibitemOpen
  \bibfield  {author} {\bibinfo {author} {\bibfnamefont {S.}~\bibnamefont
  {Takada}}, \bibinfo {author} {\bibfnamefont {H.}~\bibnamefont {Hayakawa}},\
  and\ \bibinfo {author} {\bibfnamefont {A.}~\bibnamefont {Santos}},\
  }\bibfield  {title} {\bibinfo {title} {{Mpemba effect in inertial
  suspensions}},\ }\href {https://doi.org/10.1103/PhysRevE.103.032901}
  {\bibfield  {journal} {\bibinfo  {journal} {Phys. Rev. E}\ }\textbf {\bibinfo
  {volume} {103}},\ \bibinfo {pages} {032901} (\bibinfo {year}
  {2021})}\BibitemShut {NoStop}%
\bibitem [{\citenamefont {Patr\'on}\ \emph {et~al.}(2021)\citenamefont
  {Patr\'on}, \citenamefont {S\'anchez-Rey},\ and\ \citenamefont
  {Prados}}]{Patron.2021.PRE}%
  \BibitemOpen
  \bibfield  {author} {\bibinfo {author} {\bibfnamefont {A.}~\bibnamefont
  {Patr\'on}}, \bibinfo {author} {\bibfnamefont {B.}~\bibnamefont
  {S\'anchez-Rey}},\ and\ \bibinfo {author} {\bibfnamefont {A.}~\bibnamefont
  {Prados}},\ }\bibfield  {title} {\bibinfo {title} {{Strong nonexponential
  relaxation and memory effects in a fluid with nonlinear drag}},\ }\href
  {https://doi.org/10.1103/PhysRevE.104.064127} {\bibfield  {journal} {\bibinfo
   {journal} {Phys. Rev. E}\ }\textbf {\bibinfo {volume} {104}},\ \bibinfo
  {pages} {064127} (\bibinfo {year} {2021})}\BibitemShut {NoStop}%
\bibitem [{\citenamefont {Meg\'{\i}as}\ \emph {et~al.}(2022)\citenamefont
  {Meg\'{\i}as}, \citenamefont {Santos},\ and\ \citenamefont
  {Prados}}]{Megias.2022.PRE}%
  \BibitemOpen
  \bibfield  {author} {\bibinfo {author} {\bibfnamefont {A.}~\bibnamefont
  {Meg\'{\i}as}}, \bibinfo {author} {\bibfnamefont {A.}~\bibnamefont
  {Santos}},\ and\ \bibinfo {author} {\bibfnamefont {A.}~\bibnamefont
  {Prados}},\ }\bibfield  {title} {\bibinfo {title} {{Thermal versus entropic
  Mpemba effect in molecular gases with nonlinear drag}},\ }\href
  {https://doi.org/10.1103/PhysRevE.105.054140} {\bibfield  {journal} {\bibinfo
   {journal} {Phys. Rev. E}\ }\textbf {\bibinfo {volume} {105}},\ \bibinfo
  {pages} {054140} (\bibinfo {year} {2022})}\BibitemShut {NoStop}%
\bibitem [{\citenamefont {Chatterjee}\ \emph
  {et~al.}(2024{\natexlab{a}})\citenamefont {Chatterjee}, \citenamefont
  {Ghosh}, \citenamefont {Vadakkayil}, \citenamefont {Paul}, \citenamefont
  {Singha},\ and\ \citenamefont {Das}}]{Chatterjee.2024.PRE}%
  \BibitemOpen
  \bibfield  {author} {\bibinfo {author} {\bibfnamefont {S.}~\bibnamefont
  {Chatterjee}}, \bibinfo {author} {\bibfnamefont {S.}~\bibnamefont {Ghosh}},
  \bibinfo {author} {\bibfnamefont {N.}~\bibnamefont {Vadakkayil}}, \bibinfo
  {author} {\bibfnamefont {T.}~\bibnamefont {Paul}}, \bibinfo {author}
  {\bibfnamefont {S.~K.}\ \bibnamefont {Singha}},\ and\ \bibinfo {author}
  {\bibfnamefont {S.~K.}\ \bibnamefont {Das}},\ }\bibfield  {title} {\bibinfo
  {title} {{Mpemba effect in pure spin systems : A universal picture of the
  role of spatial correlations at initial states}},\ }\href
  {https://doi.org/10.1103/PhysRevE.110.L012103} {\bibfield  {journal}
  {\bibinfo  {journal} {Phys. Rev. E}\ }\textbf {\bibinfo {volume} {110}},\
  \bibinfo {pages} {L012103} (\bibinfo {year}
  {2024}{\natexlab{a}})}\BibitemShut {NoStop}%
\bibitem [{\citenamefont {Lu}\ and\ \citenamefont {Raz}(2017)}]{Lu.2017.PNAS}%
  \BibitemOpen
  \bibfield  {author} {\bibinfo {author} {\bibfnamefont {Z.}~\bibnamefont
  {Lu}}\ and\ \bibinfo {author} {\bibfnamefont {O.}~\bibnamefont {Raz}},\
  }\bibfield  {title} {\bibinfo {title} {{Nonequilibrium thermodynamics of the
  Markovian Mpemba effect and its inverse}},\ }\href
  {https://doi.org/10.1073/pnas.1701264114} {\bibfield  {journal} {\bibinfo
  {journal} {Proc. Natl. Acad. Sci. U.S.A.}\ }\textbf {\bibinfo {volume}
  {114}},\ \bibinfo {pages} {5083} (\bibinfo {year} {2017})}\BibitemShut
  {NoStop}%
\bibitem [{\citenamefont {Nava}\ and\ \citenamefont
  {Fabrizio}(2019)}]{Nava.2019.PRB}%
  \BibitemOpen
  \bibfield  {author} {\bibinfo {author} {\bibfnamefont {A.}~\bibnamefont
  {Nava}}\ and\ \bibinfo {author} {\bibfnamefont {M.}~\bibnamefont
  {Fabrizio}},\ }\bibfield  {title} {\bibinfo {title} {{Lindblad dissipative
  dynamics in the presence of phase coexistence}},\ }\href
  {https://doi.org/10.1103/PhysRevB.100.125102} {\bibfield  {journal} {\bibinfo
   {journal} {Phys. Rev. B}\ }\textbf {\bibinfo {volume} {100}},\ \bibinfo
  {pages} {125102} (\bibinfo {year} {2019})}\BibitemShut {NoStop}%
\bibitem [{\citenamefont {Klich}\ \emph {et~al.}(2019)\citenamefont {Klich},
  \citenamefont {Raz}, \citenamefont {Hirschberg},\ and\ \citenamefont
  {Vucelja}}]{Klich.2019.PRX}%
  \BibitemOpen
  \bibfield  {author} {\bibinfo {author} {\bibfnamefont {I.}~\bibnamefont
  {Klich}}, \bibinfo {author} {\bibfnamefont {O.}~\bibnamefont {Raz}}, \bibinfo
  {author} {\bibfnamefont {O.}~\bibnamefont {Hirschberg}},\ and\ \bibinfo
  {author} {\bibfnamefont {M.}~\bibnamefont {Vucelja}},\ }\bibfield  {title}
  {\bibinfo {title} {{Mpemba index and anomalous relaxation}},\ }\href
  {https://doi.org/10.1103/PhysRevX.9.021060} {\bibfield  {journal} {\bibinfo
  {journal} {Phys. Rev. X}\ }\textbf {\bibinfo {volume} {9}},\ \bibinfo {pages}
  {021060} (\bibinfo {year} {2019})}\BibitemShut {NoStop}%
\bibitem [{\citenamefont {Gal}\ and\ \citenamefont {Raz}(2020)}]{Gal.2020.PRL}%
  \BibitemOpen
  \bibfield  {author} {\bibinfo {author} {\bibfnamefont {A.}~\bibnamefont
  {Gal}}\ and\ \bibinfo {author} {\bibfnamefont {O.}~\bibnamefont {Raz}},\
  }\bibfield  {title} {\bibinfo {title} {{Precooling strategy allows
  exponentially faster heating}},\ }\href
  {https://doi.org/10.1103/PhysRevLett.124.060602} {\bibfield  {journal}
  {\bibinfo  {journal} {Phys. Rev. Lett.}\ }\textbf {\bibinfo {volume} {124}},\
  \bibinfo {pages} {060602} (\bibinfo {year} {2020})}\BibitemShut {NoStop}%
\bibitem [{\citenamefont {Walker}\ and\ \citenamefont
  {Vucelja}(2021)}]{Walker.2021.JSM}%
  \BibitemOpen
  \bibfield  {author} {\bibinfo {author} {\bibfnamefont {M.~R.}\ \bibnamefont
  {Walker}}\ and\ \bibinfo {author} {\bibfnamefont {M.}~\bibnamefont
  {Vucelja}},\ }\bibfield  {title} {\bibinfo {title} {{Anomalous thermal
  relaxation of Langevin particles in a piecewise-constant potential}},\ }\href
  {https://doi.org/10.1088/1742-5468/ac2edc} {\bibfield  {journal} {\bibinfo
  {journal} {J. Stat. Mech.}\ }\textbf {\bibinfo {volume} {2021}},\ \bibinfo
  {pages} {113105} (\bibinfo {year} {2021})}\BibitemShut {NoStop}%
\bibitem [{\citenamefont {Carollo}\ \emph {et~al.}(2021)\citenamefont
  {Carollo}, \citenamefont {Lasanta},\ and\ \citenamefont
  {Lesanovsky}}]{Carollo.2021.PRL}%
  \BibitemOpen
  \bibfield  {author} {\bibinfo {author} {\bibfnamefont {F.}~\bibnamefont
  {Carollo}}, \bibinfo {author} {\bibfnamefont {A.}~\bibnamefont {Lasanta}},\
  and\ \bibinfo {author} {\bibfnamefont {I.}~\bibnamefont {Lesanovsky}},\
  }\bibfield  {title} {\bibinfo {title} {{Exponentially accelerated approach to
  stationarity in Markovian open quantum systems through the Mpemba effect}},\
  }\href {https://doi.org/10.1103/PhysRevLett.127.060401} {\bibfield  {journal}
  {\bibinfo  {journal} {Phys. Rev. Lett.}\ }\textbf {\bibinfo {volume} {127}},\
  \bibinfo {pages} {060401} (\bibinfo {year} {2021})}\BibitemShut {NoStop}%
\bibitem [{\citenamefont {Busiello}\ \emph {et~al.}(2021)\citenamefont
  {Busiello}, \citenamefont {Gupta},\ and\ \citenamefont
  {Maritan}}]{Busiello.2021.NJP}%
  \BibitemOpen
  \bibfield  {author} {\bibinfo {author} {\bibfnamefont {D.~M.}\ \bibnamefont
  {Busiello}}, \bibinfo {author} {\bibfnamefont {D.}~\bibnamefont {Gupta}},\
  and\ \bibinfo {author} {\bibfnamefont {A.}~\bibnamefont {Maritan}},\
  }\bibfield  {title} {\bibinfo {title} {{Inducing and optimizing Markovian
  Mpemba effect with stochastic reset}},\ }\href
  {https://doi.org/10.1088/1367-2630/ac2922} {\bibfield  {journal} {\bibinfo
  {journal} {New J. Phys.}\ }\textbf {\bibinfo {volume} {23}},\ \bibinfo
  {pages} {103012} (\bibinfo {year} {2021})}\BibitemShut {NoStop}%
\bibitem [{\citenamefont {Deg\"{u}nther}\ and\ \citenamefont
  {Seifert}(2022)}]{Degnther.2022.EPL}%
  \BibitemOpen
  \bibfield  {author} {\bibinfo {author} {\bibfnamefont {J.}~\bibnamefont
  {Deg\"{u}nther}}\ and\ \bibinfo {author} {\bibfnamefont {U.}~\bibnamefont
  {Seifert}},\ }\bibfield  {title} {\bibinfo {title} {{Anomalous relaxation
  from a non-equilibrium steady state: An isothermal analog of the Mpemba
  effect}},\ }\href {https://doi.org/10.1209/0295-5075/ac8573} {\bibfield
  {journal} {\bibinfo  {journal} {Europhys. Lett.}\ }\textbf {\bibinfo {volume}
  {139}},\ \bibinfo {pages} {41002} (\bibinfo {year} {2022})}\BibitemShut
  {NoStop}%
\bibitem [{\citenamefont {Kochsiek}\ \emph {et~al.}(2022)\citenamefont
  {Kochsiek}, \citenamefont {Carollo},\ and\ \citenamefont
  {Lesanovsky}}]{Kochsiek.2022.PRA}%
  \BibitemOpen
  \bibfield  {author} {\bibinfo {author} {\bibfnamefont {S.}~\bibnamefont
  {Kochsiek}}, \bibinfo {author} {\bibfnamefont {F.}~\bibnamefont {Carollo}},\
  and\ \bibinfo {author} {\bibfnamefont {I.}~\bibnamefont {Lesanovsky}},\
  }\bibfield  {title} {\bibinfo {title} {{Accelerating the approach of
  dissipative quantum spin systems towards stationarity through global spin
  rotations}},\ }\href {https://doi.org/10.1103/PhysRevA.106.012207} {\bibfield
   {journal} {\bibinfo  {journal} {Phys. Rev. A}\ }\textbf {\bibinfo {volume}
  {106}},\ \bibinfo {pages} {012207} (\bibinfo {year} {2022})}\BibitemShut
  {NoStop}%
\bibitem [{\citenamefont {Schwarzendahl}\ and\ \citenamefont
  {L\"owen}(2022)}]{Schwarzendahl.2022.PRL}%
  \BibitemOpen
  \bibfield  {author} {\bibinfo {author} {\bibfnamefont {F.~J.}\ \bibnamefont
  {Schwarzendahl}}\ and\ \bibinfo {author} {\bibfnamefont {H.}~\bibnamefont
  {L\"owen}},\ }\bibfield  {title} {\bibinfo {title} {{Anomalous cooling and
  overcooling of active colloids}},\ }\href
  {https://doi.org/10.1103/PhysRevLett.129.138002} {\bibfield  {journal}
  {\bibinfo  {journal} {Phys. Rev. Lett.}\ }\textbf {\bibinfo {volume} {129}},\
  \bibinfo {pages} {138002} (\bibinfo {year} {2022})}\BibitemShut {NoStop}%
\bibitem [{\citenamefont {Teza}\ \emph {et~al.}(2023)\citenamefont {Teza},
  \citenamefont {Yaacoby},\ and\ \citenamefont {Raz}}]{Teza.2023.PRL}%
  \BibitemOpen
  \bibfield  {author} {\bibinfo {author} {\bibfnamefont {G.}~\bibnamefont
  {Teza}}, \bibinfo {author} {\bibfnamefont {R.}~\bibnamefont {Yaacoby}},\ and\
  \bibinfo {author} {\bibfnamefont {O.}~\bibnamefont {Raz}},\ }\bibfield
  {title} {\bibinfo {title} {{Relaxation shortcuts through boundary
  coupling}},\ }\href {https://doi.org/10.1103/PhysRevLett.131.017101}
  {\bibfield  {journal} {\bibinfo  {journal} {Phys. Rev. Lett.}\ }\textbf
  {\bibinfo {volume} {131}},\ \bibinfo {pages} {017101} (\bibinfo {year}
  {2023})}\BibitemShut {NoStop}%
\bibitem [{\citenamefont {Chatterjee}\ \emph {et~al.}(2023)\citenamefont
  {Chatterjee}, \citenamefont {Takada},\ and\ \citenamefont
  {Hayakawa}}]{Chatterjee.2023.PRL}%
  \BibitemOpen
  \bibfield  {author} {\bibinfo {author} {\bibfnamefont {A.~K.}\ \bibnamefont
  {Chatterjee}}, \bibinfo {author} {\bibfnamefont {S.}~\bibnamefont {Takada}},\
  and\ \bibinfo {author} {\bibfnamefont {H.}~\bibnamefont {Hayakawa}},\
  }\bibfield  {title} {\bibinfo {title} {{Quantum Mpemba effect in a quantum
  dot with reservoirs}},\ }\href
  {https://doi.org/10.1103/PhysRevLett.131.080402} {\bibfield  {journal}
  {\bibinfo  {journal} {Phys. Rev. Lett.}\ }\textbf {\bibinfo {volume} {131}},\
  \bibinfo {pages} {080402} (\bibinfo {year} {2023})}\BibitemShut {NoStop}%
\bibitem [{\citenamefont {Ivander}\ \emph {et~al.}(2023)\citenamefont
  {Ivander}, \citenamefont {Anto-Sztrikacs},\ and\ \citenamefont
  {Segal}}]{Ivander.2023.PRE}%
  \BibitemOpen
  \bibfield  {author} {\bibinfo {author} {\bibfnamefont {F.}~\bibnamefont
  {Ivander}}, \bibinfo {author} {\bibfnamefont {N.}~\bibnamefont
  {Anto-Sztrikacs}},\ and\ \bibinfo {author} {\bibfnamefont {D.}~\bibnamefont
  {Segal}},\ }\bibfield  {title} {\bibinfo {title} {{Hyperacceleration of
  quantum thermalization dynamics by bypassing long-lived coherences: An
  analytical treatment}},\ }\href {https://doi.org/10.1103/PhysRevE.108.014130}
  {\bibfield  {journal} {\bibinfo  {journal} {Phys. Rev. E}\ }\textbf {\bibinfo
  {volume} {108}},\ \bibinfo {pages} {014130} (\bibinfo {year}
  {2023})}\BibitemShut {NoStop}%
\bibitem [{\citenamefont {Walker}\ \emph {et~al.}(2023)\citenamefont {Walker},
  \citenamefont {Bera},\ and\ \citenamefont {Vucelja}}]{Walker.2023.arxiv}%
  \BibitemOpen
  \bibfield  {author} {\bibinfo {author} {\bibfnamefont {M.~R.}\ \bibnamefont
  {Walker}}, \bibinfo {author} {\bibfnamefont {S.}~\bibnamefont {Bera}},\ and\
  \bibinfo {author} {\bibfnamefont {M.}~\bibnamefont {Vucelja}},\ }\bibfield
  {title} {\bibinfo {title} {{Optimal transport and anomalous thermal
  relaxations}},\ }\href {https://arxiv.org/abs/2307.16103} {\bibfield
  {journal} {\bibinfo  {journal} {arXiv:2307.16103}\ } (\bibinfo {year}
  {2023})}\BibitemShut {NoStop}%
\bibitem [{\citenamefont {Chatterjee}\ \emph
  {et~al.}(2024{\natexlab{b}})\citenamefont {Chatterjee}, \citenamefont
  {Takada},\ and\ \citenamefont {Hayakawa}}]{Chatterjee.2024.PRA}%
  \BibitemOpen
  \bibfield  {author} {\bibinfo {author} {\bibfnamefont {A.~K.}\ \bibnamefont
  {Chatterjee}}, \bibinfo {author} {\bibfnamefont {S.}~\bibnamefont {Takada}},\
  and\ \bibinfo {author} {\bibfnamefont {H.}~\bibnamefont {Hayakawa}},\
  }\bibfield  {title} {\bibinfo {title} {{Multiple quantum Mpemba effect:
  Exceptional points and oscillations}},\ }\href
  {https://doi.org/10.1103/PhysRevA.110.022213} {\bibfield  {journal} {\bibinfo
   {journal} {Phys. Rev. A}\ }\textbf {\bibinfo {volume} {110}},\ \bibinfo
  {pages} {022213} (\bibinfo {year} {2024}{\natexlab{b}})}\BibitemShut
  {NoStop}%
\bibitem [{\citenamefont {Wang}\ and\ \citenamefont
  {Wang}(2024)}]{Wang.2024.PRR}%
  \BibitemOpen
  \bibfield  {author} {\bibinfo {author} {\bibfnamefont {X.}~\bibnamefont
  {Wang}}\ and\ \bibinfo {author} {\bibfnamefont {J.}~\bibnamefont {Wang}},\
  }\bibfield  {title} {\bibinfo {title} {{Mpemba effects in nonequilibrium open
  quantum systems}},\ }\href {https://doi.org/10.1103/PhysRevResearch.6.033330}
  {\bibfield  {journal} {\bibinfo  {journal} {Phys. Rev. Res.}\ }\textbf
  {\bibinfo {volume} {6}},\ \bibinfo {pages} {033330} (\bibinfo {year}
  {2024})}\BibitemShut {NoStop}%
\bibitem [{\citenamefont {Moroder}\ \emph {et~al.}(2024)\citenamefont
  {Moroder}, \citenamefont {Culhane}, \citenamefont {Zawadzki},\ and\
  \citenamefont {Goold}}]{Moroder.2024.PRL}%
  \BibitemOpen
  \bibfield  {author} {\bibinfo {author} {\bibfnamefont {M.}~\bibnamefont
  {Moroder}}, \bibinfo {author} {\bibfnamefont {O.}~\bibnamefont {Culhane}},
  \bibinfo {author} {\bibfnamefont {K.}~\bibnamefont {Zawadzki}},\ and\
  \bibinfo {author} {\bibfnamefont {J.}~\bibnamefont {Goold}},\ }\bibfield
  {title} {\bibinfo {title} {{Thermodynamics of the quantum Mpemba effect}},\
  }\href {https://doi.org/10.1103/PhysRevLett.133.140404} {\bibfield  {journal}
  {\bibinfo  {journal} {Phys. Rev. Lett.}\ }\textbf {\bibinfo {volume} {133}},\
  \bibinfo {pages} {140404} (\bibinfo {year} {2024})}\BibitemShut {NoStop}%
\bibitem [{\citenamefont {Pemart\'{\i}n}\ \emph {et~al.}(2024)\citenamefont
  {Pemart\'{\i}n}, \citenamefont {Momp\'o}, \citenamefont {Lasanta},
  \citenamefont {Mart\'{\i}n-Mayor},\ and\ \citenamefont
  {Salas}}]{Pemartin.2024.PRL}%
  \BibitemOpen
  \bibfield  {author} {\bibinfo {author} {\bibfnamefont {I.~G.-A.}\
  \bibnamefont {Pemart\'{\i}n}}, \bibinfo {author} {\bibfnamefont
  {E.}~\bibnamefont {Momp\'o}}, \bibinfo {author} {\bibfnamefont
  {A.}~\bibnamefont {Lasanta}}, \bibinfo {author} {\bibfnamefont
  {V.}~\bibnamefont {Mart\'{\i}n-Mayor}},\ and\ \bibinfo {author}
  {\bibfnamefont {J.}~\bibnamefont {Salas}},\ }\bibfield  {title} {\bibinfo
  {title} {{Shortcuts of freely relaxing systems using equilibrium physical
  observables}},\ }\href {https://doi.org/10.1103/PhysRevLett.132.117102}
  {\bibfield  {journal} {\bibinfo  {journal} {Phys. Rev. Lett.}\ }\textbf
  {\bibinfo {volume} {132}},\ \bibinfo {pages} {117102} (\bibinfo {year}
  {2024})}\BibitemShut {NoStop}%
\bibitem [{\citenamefont {Nava}\ and\ \citenamefont
  {Egger}(2024)}]{Nava.2024.PRL}%
  \BibitemOpen
  \bibfield  {author} {\bibinfo {author} {\bibfnamefont {A.}~\bibnamefont
  {Nava}}\ and\ \bibinfo {author} {\bibfnamefont {R.}~\bibnamefont {Egger}},\
  }\bibfield  {title} {\bibinfo {title} {{Mpemba effects in open nonequilibrium
  quantum systems}},\ }\href {https://doi.org/10.1103/PhysRevLett.133.136302}
  {\bibfield  {journal} {\bibinfo  {journal} {Phys. Rev. Lett.}\ }\textbf
  {\bibinfo {volume} {133}},\ \bibinfo {pages} {136302} (\bibinfo {year}
  {2024})}\BibitemShut {NoStop}%
\bibitem [{\citenamefont {Strachan}\ \emph {et~al.}(2024)\citenamefont
  {Strachan}, \citenamefont {Purkayastha},\ and\ \citenamefont
  {Clark}}]{Strachan.2024.arxiv}%
  \BibitemOpen
  \bibfield  {author} {\bibinfo {author} {\bibfnamefont {D.~J.}\ \bibnamefont
  {Strachan}}, \bibinfo {author} {\bibfnamefont {A.}~\bibnamefont
  {Purkayastha}},\ and\ \bibinfo {author} {\bibfnamefont {S.~R.}\ \bibnamefont
  {Clark}},\ }\bibfield  {title} {\bibinfo {title} {{Non-Markovian quantum
  Mpemba effect}},\ }\href {https://arxiv.org/abs/2402.05756} {\bibfield
  {journal} {\bibinfo  {journal} {arXiv:2402.05756}\ } (\bibinfo {year}
  {2024})}\BibitemShut {NoStop}%
\bibitem [{\citenamefont {Aharony~Shapira}\ \emph {et~al.}(2024)\citenamefont
  {Aharony~Shapira}, \citenamefont {Shapira}, \citenamefont {Markov},
  \citenamefont {Teza}, \citenamefont {Akerman}, \citenamefont {Raz},\ and\
  \citenamefont {Ozeri}}]{Shapira.2024.PRL}%
  \BibitemOpen
  \bibfield  {author} {\bibinfo {author} {\bibfnamefont {S.}~\bibnamefont
  {Aharony~Shapira}}, \bibinfo {author} {\bibfnamefont {Y.}~\bibnamefont
  {Shapira}}, \bibinfo {author} {\bibfnamefont {J.}~\bibnamefont {Markov}},
  \bibinfo {author} {\bibfnamefont {G.}~\bibnamefont {Teza}}, \bibinfo {author}
  {\bibfnamefont {N.}~\bibnamefont {Akerman}}, \bibinfo {author} {\bibfnamefont
  {O.}~\bibnamefont {Raz}},\ and\ \bibinfo {author} {\bibfnamefont
  {R.}~\bibnamefont {Ozeri}},\ }\bibfield  {title} {\bibinfo {title} {{Inverse
  Mpemba effect demonstrated on a single trapped ion qubit}},\ }\href
  {https://doi.org/10.1103/PhysRevLett.133.010403} {\bibfield  {journal}
  {\bibinfo  {journal} {Phys. Rev. Lett.}\ }\textbf {\bibinfo {volume} {133}},\
  \bibinfo {pages} {010403} (\bibinfo {year} {2024})}\BibitemShut {NoStop}%
\bibitem [{\citenamefont {Lapolla}\ and\ \citenamefont
  {Godec}(2020)}]{Lapolla.2020.PRL}%
  \BibitemOpen
  \bibfield  {author} {\bibinfo {author} {\bibfnamefont {A.}~\bibnamefont
  {Lapolla}}\ and\ \bibinfo {author} {\bibfnamefont {A.~c.~v.}\ \bibnamefont
  {Godec}},\ }\bibfield  {title} {\bibinfo {title} {{Faster uphill relaxation
  in thermodynamically equidistant temperature quenches}},\ }\href
  {https://doi.org/10.1103/PhysRevLett.125.110602} {\bibfield  {journal}
  {\bibinfo  {journal} {Phys. Rev. Lett.}\ }\textbf {\bibinfo {volume} {125}},\
  \bibinfo {pages} {110602} (\bibinfo {year} {2020})}\BibitemShut {NoStop}%
\bibitem [{\citenamefont {Van~Vu}\ and\ \citenamefont
  {Hasegawa}(2021)}]{Vu.2021.PRR}%
  \BibitemOpen
  \bibfield  {author} {\bibinfo {author} {\bibfnamefont {T.}~\bibnamefont
  {Van~Vu}}\ and\ \bibinfo {author} {\bibfnamefont {Y.}~\bibnamefont
  {Hasegawa}},\ }\bibfield  {title} {\bibinfo {title} {{Toward relaxation
  asymmetry: Heating is faster than cooling}},\ }\href
  {https://doi.org/10.1103/PhysRevResearch.3.043160} {\bibfield  {journal}
  {\bibinfo  {journal} {Phys. Rev. Res.}\ }\textbf {\bibinfo {volume} {3}},\
  \bibinfo {pages} {043160} (\bibinfo {year} {2021})}\BibitemShut {NoStop}%
\bibitem [{\citenamefont {Manikandan}(2021)}]{Manikandan.2021.PRR}%
  \BibitemOpen
  \bibfield  {author} {\bibinfo {author} {\bibfnamefont {S.~K.}\ \bibnamefont
  {Manikandan}},\ }\bibfield  {title} {\bibinfo {title} {{Equidistant quenches
  in few-level quantum systems}},\ }\href
  {https://doi.org/10.1103/PhysRevResearch.3.043108} {\bibfield  {journal}
  {\bibinfo  {journal} {Phys. Rev. Res.}\ }\textbf {\bibinfo {volume} {3}},\
  \bibinfo {pages} {043108} (\bibinfo {year} {2021})}\BibitemShut {NoStop}%
\bibitem [{\citenamefont {Meibohm}\ \emph {et~al.}(2021)\citenamefont
  {Meibohm}, \citenamefont {Forastiere}, \citenamefont {Adeleke-Larodo},\ and\
  \citenamefont {Proesmans}}]{Meibohm.2021.PRE}%
  \BibitemOpen
  \bibfield  {author} {\bibinfo {author} {\bibfnamefont {J.}~\bibnamefont
  {Meibohm}}, \bibinfo {author} {\bibfnamefont {D.}~\bibnamefont {Forastiere}},
  \bibinfo {author} {\bibfnamefont {T.}~\bibnamefont {Adeleke-Larodo}},\ and\
  \bibinfo {author} {\bibfnamefont {K.}~\bibnamefont {Proesmans}},\ }\bibfield
  {title} {\bibinfo {title} {{Relaxation-speed crossover in anharmonic
  potentials}},\ }\href {https://doi.org/10.1103/PhysRevE.104.L032105}
  {\bibfield  {journal} {\bibinfo  {journal} {Phys. Rev. E}\ }\textbf {\bibinfo
  {volume} {104}},\ \bibinfo {pages} {L032105} (\bibinfo {year}
  {2021})}\BibitemShut {NoStop}%
\bibitem [{\citenamefont {Ibáñez}\ \emph {et~al.}(2024)\citenamefont
  {Ibáñez}, \citenamefont {Dieball}, \citenamefont {Lasanta}, \citenamefont
  {Godec},\ and\ \citenamefont {Rica}}]{Ibez.2024.NP}%
  \BibitemOpen
  \bibfield  {author} {\bibinfo {author} {\bibfnamefont {M.}~\bibnamefont
  {Ibáñez}}, \bibinfo {author} {\bibfnamefont {C.}~\bibnamefont {Dieball}},
  \bibinfo {author} {\bibfnamefont {A.}~\bibnamefont {Lasanta}}, \bibinfo
  {author} {\bibfnamefont {A.}~\bibnamefont {Godec}},\ and\ \bibinfo {author}
  {\bibfnamefont {R.~A.}\ \bibnamefont {Rica}},\ }\bibfield  {title} {\bibinfo
  {title} {{Heating and cooling are fundamentally asymmetric and evolve along
  distinct pathways}},\ }\href {https://doi.org/10.1038/s41567-023-02269-z}
  {\bibfield  {journal} {\bibinfo  {journal} {Nat. Phys.}\ }\textbf {\bibinfo
  {volume} {20}},\ \bibinfo {pages} {135} (\bibinfo {year} {2024})}\BibitemShut
  {NoStop}%
\bibitem [{\citenamefont {Ares}\ \emph {et~al.}(2023)\citenamefont {Ares},
  \citenamefont {Murciano},\ and\ \citenamefont {Calabrese}}]{Ares.2023.NC}%
  \BibitemOpen
  \bibfield  {author} {\bibinfo {author} {\bibfnamefont {F.}~\bibnamefont
  {Ares}}, \bibinfo {author} {\bibfnamefont {S.}~\bibnamefont {Murciano}},\
  and\ \bibinfo {author} {\bibfnamefont {P.}~\bibnamefont {Calabrese}},\
  }\bibfield  {title} {\bibinfo {title} {{Entanglement asymmetry as a probe of
  symmetry breaking}},\ }\href {https://doi.org/10.1038/s41467-023-37747-8}
  {\bibfield  {journal} {\bibinfo  {journal} {Nat. Commun.}\ }\textbf {\bibinfo
  {volume} {14}},\ \bibinfo {pages} {2036} (\bibinfo {year}
  {2023})}\BibitemShut {NoStop}%
\bibitem [{\citenamefont {Murciano}\ \emph {et~al.}(2024)\citenamefont
  {Murciano}, \citenamefont {Ares}, \citenamefont {Klich},\ and\ \citenamefont
  {Calabrese}}]{Murciano.2024.JSM}%
  \BibitemOpen
  \bibfield  {author} {\bibinfo {author} {\bibfnamefont {S.}~\bibnamefont
  {Murciano}}, \bibinfo {author} {\bibfnamefont {F.}~\bibnamefont {Ares}},
  \bibinfo {author} {\bibfnamefont {I.}~\bibnamefont {Klich}},\ and\ \bibinfo
  {author} {\bibfnamefont {P.}~\bibnamefont {Calabrese}},\ }\bibfield  {title}
  {\bibinfo {title} {{Entanglement asymmetry and quantum Mpemba effect in the
  XY spin chain}},\ }\href {https://doi.org/10.1088/1742-5468/ad17b4}
  {\bibfield  {journal} {\bibinfo  {journal} {J. Stat. Mech.}\ }\textbf
  {\bibinfo {volume} {2024}},\ \bibinfo {pages} {013103} (\bibinfo {year}
  {2024})}\BibitemShut {NoStop}%
\bibitem [{\citenamefont {Rylands}\ \emph {et~al.}(2024)\citenamefont
  {Rylands}, \citenamefont {Klobas}, \citenamefont {Ares}, \citenamefont
  {Calabrese}, \citenamefont {Murciano},\ and\ \citenamefont
  {Bertini}}]{Rylands.2024.PRL}%
  \BibitemOpen
  \bibfield  {author} {\bibinfo {author} {\bibfnamefont {C.}~\bibnamefont
  {Rylands}}, \bibinfo {author} {\bibfnamefont {K.}~\bibnamefont {Klobas}},
  \bibinfo {author} {\bibfnamefont {F.}~\bibnamefont {Ares}}, \bibinfo {author}
  {\bibfnamefont {P.}~\bibnamefont {Calabrese}}, \bibinfo {author}
  {\bibfnamefont {S.}~\bibnamefont {Murciano}},\ and\ \bibinfo {author}
  {\bibfnamefont {B.}~\bibnamefont {Bertini}},\ }\bibfield  {title} {\bibinfo
  {title} {{Microscopic origin of the quantum Mpemba effect in integrable
  systems}},\ }\href {https://doi.org/10.1103/PhysRevLett.133.010401}
  {\bibfield  {journal} {\bibinfo  {journal} {Phys. Rev. Lett.}\ }\textbf
  {\bibinfo {volume} {133}},\ \bibinfo {pages} {010401} (\bibinfo {year}
  {2024})}\BibitemShut {NoStop}%
\bibitem [{\citenamefont {Liu}\ \emph {et~al.}(2024)\citenamefont {Liu},
  \citenamefont {Zhang}, \citenamefont {Yin},\ and\ \citenamefont
  {Zhang}}]{Liu.2024.PRL}%
  \BibitemOpen
  \bibfield  {author} {\bibinfo {author} {\bibfnamefont {S.}~\bibnamefont
  {Liu}}, \bibinfo {author} {\bibfnamefont {H.-K.}\ \bibnamefont {Zhang}},
  \bibinfo {author} {\bibfnamefont {S.}~\bibnamefont {Yin}},\ and\ \bibinfo
  {author} {\bibfnamefont {S.-X.}\ \bibnamefont {Zhang}},\ }\bibfield  {title}
  {\bibinfo {title} {{Symmetry restoration and quantum Mpemba effect in
  symmetric random circuits}},\ }\href
  {https://doi.org/10.1103/PhysRevLett.133.140405} {\bibfield  {journal}
  {\bibinfo  {journal} {Phys. Rev. Lett.}\ }\textbf {\bibinfo {volume} {133}},\
  \bibinfo {pages} {140405} (\bibinfo {year} {2024})}\BibitemShut {NoStop}%
\bibitem [{\citenamefont {Yamashika}\ \emph {et~al.}(2024)\citenamefont
  {Yamashika}, \citenamefont {Ares},\ and\ \citenamefont
  {Calabrese}}]{Yamashika.2024.PRB}%
  \BibitemOpen
  \bibfield  {author} {\bibinfo {author} {\bibfnamefont {S.}~\bibnamefont
  {Yamashika}}, \bibinfo {author} {\bibfnamefont {F.}~\bibnamefont {Ares}},\
  and\ \bibinfo {author} {\bibfnamefont {P.}~\bibnamefont {Calabrese}},\
  }\bibfield  {title} {\bibinfo {title} {{Entanglement asymmetry and quantum
  Mpemba effect in two-dimensional free-fermion systems}},\ }\href
  {https://doi.org/10.1103/PhysRevB.110.085126} {\bibfield  {journal} {\bibinfo
   {journal} {Phys. Rev. B}\ }\textbf {\bibinfo {volume} {110}},\ \bibinfo
  {pages} {085126} (\bibinfo {year} {2024})}\BibitemShut {NoStop}%
\bibitem [{\citenamefont {Chang}\ \emph {et~al.}(2024)\citenamefont {Chang},
  \citenamefont {Yin}, \citenamefont {Zhang},\ and\ \citenamefont
  {Li}}]{Chang.2024.arxiv}%
  \BibitemOpen
  \bibfield  {author} {\bibinfo {author} {\bibfnamefont {W.-X.}\ \bibnamefont
  {Chang}}, \bibinfo {author} {\bibfnamefont {S.}~\bibnamefont {Yin}}, \bibinfo
  {author} {\bibfnamefont {S.-X.}\ \bibnamefont {Zhang}},\ and\ \bibinfo
  {author} {\bibfnamefont {Z.-X.}\ \bibnamefont {Li}},\ }\bibfield  {title}
  {\bibinfo {title} {{Imaginary-time Mpemba effect in quantum many-body
  systems}},\ }\href {https://arxiv.org/abs/2409.06547} {\bibfield  {journal}
  {\bibinfo  {journal} {arXiv:2409.06547}\ } (\bibinfo {year}
  {2024})}\BibitemShut {NoStop}%
\bibitem [{\citenamefont {Kumar}\ and\ \citenamefont
  {Bechhoefer}(2020)}]{Kumar.2020.N}%
  \BibitemOpen
  \bibfield  {author} {\bibinfo {author} {\bibfnamefont {A.}~\bibnamefont
  {Kumar}}\ and\ \bibinfo {author} {\bibfnamefont {J.}~\bibnamefont
  {Bechhoefer}},\ }\bibfield  {title} {\bibinfo {title} {{Exponentially faster
  cooling in a colloidal system}},\ }\href
  {https://doi.org/10.1038/s41586-020-2560-x} {\bibfield  {journal} {\bibinfo
  {journal} {Nature}\ }\textbf {\bibinfo {volume} {584}},\ \bibinfo {pages}
  {64} (\bibinfo {year} {2020})}\BibitemShut {NoStop}%
\bibitem [{\citenamefont {Kumar}\ \emph {et~al.}(2022)\citenamefont {Kumar},
  \citenamefont {Chétrite},\ and\ \citenamefont
  {Bechhoefer}}]{Kumar.2022.PNAS}%
  \BibitemOpen
  \bibfield  {author} {\bibinfo {author} {\bibfnamefont {A.}~\bibnamefont
  {Kumar}}, \bibinfo {author} {\bibfnamefont {R.}~\bibnamefont {Chétrite}},\
  and\ \bibinfo {author} {\bibfnamefont {J.}~\bibnamefont {Bechhoefer}},\
  }\bibfield  {title} {\bibinfo {title} {{Anomalous heating in a colloidal
  system}},\ }\href {https://doi.org/10.1073/pnas.2118484119} {\bibfield
  {journal} {\bibinfo  {journal} {Proc. Natl. Acad. Sci. U.S.A.}\ }\textbf
  {\bibinfo {volume} {119}} (\bibinfo {year} {2022})}\BibitemShut {NoStop}%
\bibitem [{\citenamefont {Joshi}\ \emph {et~al.}(2024)\citenamefont {Joshi},
  \citenamefont {Franke}, \citenamefont {Rath}, \citenamefont {Ares},
  \citenamefont {Murciano}, \citenamefont {Kranzl}, \citenamefont {Blatt},
  \citenamefont {Zoller}, \citenamefont {Vermersch}, \citenamefont {Calabrese},
  \citenamefont {Roos},\ and\ \citenamefont {Joshi}}]{Joshi.2024.PRL}%
  \BibitemOpen
  \bibfield  {author} {\bibinfo {author} {\bibfnamefont {L.~K.}\ \bibnamefont
  {Joshi}}, \bibinfo {author} {\bibfnamefont {J.}~\bibnamefont {Franke}},
  \bibinfo {author} {\bibfnamefont {A.}~\bibnamefont {Rath}}, \bibinfo {author}
  {\bibfnamefont {F.}~\bibnamefont {Ares}}, \bibinfo {author} {\bibfnamefont
  {S.}~\bibnamefont {Murciano}}, \bibinfo {author} {\bibfnamefont
  {F.}~\bibnamefont {Kranzl}}, \bibinfo {author} {\bibfnamefont
  {R.}~\bibnamefont {Blatt}}, \bibinfo {author} {\bibfnamefont
  {P.}~\bibnamefont {Zoller}}, \bibinfo {author} {\bibfnamefont
  {B.}~\bibnamefont {Vermersch}}, \bibinfo {author} {\bibfnamefont
  {P.}~\bibnamefont {Calabrese}}, \bibinfo {author} {\bibfnamefont {C.~F.}\
  \bibnamefont {Roos}},\ and\ \bibinfo {author} {\bibfnamefont {M.~K.}\
  \bibnamefont {Joshi}},\ }\bibfield  {title} {\bibinfo {title} {{Observing the
  quantum Mpemba effect in quantum simulations}},\ }\href
  {https://doi.org/10.1103/PhysRevLett.133.010402} {\bibfield  {journal}
  {\bibinfo  {journal} {Phys. Rev. Lett.}\ }\textbf {\bibinfo {volume} {133}},\
  \bibinfo {pages} {010402} (\bibinfo {year} {2024})}\BibitemShut {NoStop}%
\bibitem [{\citenamefont {Biswas}\ \emph {et~al.}(2023)\citenamefont {Biswas},
  \citenamefont {Prasad},\ and\ \citenamefont {Rajesh}}]{Biswas.2023.PRE}%
  \BibitemOpen
  \bibfield  {author} {\bibinfo {author} {\bibfnamefont {A.}~\bibnamefont
  {Biswas}}, \bibinfo {author} {\bibfnamefont {V.~V.}\ \bibnamefont {Prasad}},\
  and\ \bibinfo {author} {\bibfnamefont {R.}~\bibnamefont {Rajesh}},\
  }\bibfield  {title} {\bibinfo {title} {{Mpemba effect in driven granular
  gases: Role of distance measures}},\ }\href
  {https://doi.org/10.1103/PhysRevE.108.024902} {\bibfield  {journal} {\bibinfo
   {journal} {Phys. Rev. E}\ }\textbf {\bibinfo {volume} {108}},\ \bibinfo
  {pages} {024902} (\bibinfo {year} {2023})}\BibitemShut {NoStop}%
\bibitem [{\citenamefont {Ruch}\ \emph {et~al.}(1978)\citenamefont {Ruch},
  \citenamefont {Schranner},\ and\ \citenamefont {Seligman}}]{Ruch.1978.JCP}%
  \BibitemOpen
  \bibfield  {author} {\bibinfo {author} {\bibfnamefont {E.}~\bibnamefont
  {Ruch}}, \bibinfo {author} {\bibfnamefont {R.}~\bibnamefont {Schranner}},\
  and\ \bibinfo {author} {\bibfnamefont {T.~H.}\ \bibnamefont {Seligman}},\
  }\bibfield  {title} {\bibinfo {title} {{The mixing distance}},\ }\href
  {https://doi.org/10.1063/1.436364} {\bibfield  {journal} {\bibinfo  {journal}
  {J. Chem. Phys.}\ }\textbf {\bibinfo {volume} {69}},\ \bibinfo {pages} {386}
  (\bibinfo {year} {1978})}\BibitemShut {NoStop}%
\bibitem [{\citenamefont {Horodecki}\ and\ \citenamefont
  {Oppenheim}(2013)}]{Horodecki.2013.NC}%
  \BibitemOpen
  \bibfield  {author} {\bibinfo {author} {\bibfnamefont {M.}~\bibnamefont
  {Horodecki}}\ and\ \bibinfo {author} {\bibfnamefont {J.}~\bibnamefont
  {Oppenheim}},\ }\bibfield  {title} {\bibinfo {title} {{Fundamental
  limitations for quantum and nanoscale thermodynamics}},\ }\href
  {https://doi.org/10.1038/ncomms3059} {\bibfield  {journal} {\bibinfo
  {journal} {Nat. Commun.}\ }\textbf {\bibinfo {volume} {4}},\ \bibinfo {pages}
  {2059} (\bibinfo {year} {2013})}\BibitemShut {NoStop}%
\bibitem [{\citenamefont {Marshall}\ \emph {et~al.}(2011)\citenamefont
  {Marshall}, \citenamefont {Olkin},\ and\ \citenamefont
  {Arnold}}]{Marshall.2011}%
  \BibitemOpen
  \bibfield  {author} {\bibinfo {author} {\bibfnamefont {A.~W.}\ \bibnamefont
  {Marshall}}, \bibinfo {author} {\bibfnamefont {I.}~\bibnamefont {Olkin}},\
  and\ \bibinfo {author} {\bibfnamefont {B.~C.}\ \bibnamefont {Arnold}},\
  }\href@noop {} {\emph {\bibinfo {title} {{Inequalities: Theory of
  majorization and its applications}}}},\ \bibinfo {edition} {2nd}\ ed.\
  (\bibinfo  {publisher} {Springer New York, NY},\ \bibinfo {year}
  {2011})\BibitemShut {NoStop}%
\bibitem [{\citenamefont {Streltsov}\ \emph {et~al.}(2017)\citenamefont
  {Streltsov}, \citenamefont {Adesso},\ and\ \citenamefont
  {Plenio}}]{Streltsov.2017.RMP}%
  \BibitemOpen
  \bibfield  {author} {\bibinfo {author} {\bibfnamefont {A.}~\bibnamefont
  {Streltsov}}, \bibinfo {author} {\bibfnamefont {G.}~\bibnamefont {Adesso}},\
  and\ \bibinfo {author} {\bibfnamefont {M.~B.}\ \bibnamefont {Plenio}},\
  }\bibfield  {title} {\bibinfo {title} {{Colloquium: Quantum coherence as a
  resource}},\ }\href {https://doi.org/10.1103/RevModPhys.89.041003} {\bibfield
   {journal} {\bibinfo  {journal} {Rev. Mod. Phys.}\ }\textbf {\bibinfo
  {volume} {89}},\ \bibinfo {pages} {041003} (\bibinfo {year}
  {2017})}\BibitemShut {NoStop}%
\bibitem [{\citenamefont {Chitambar}\ and\ \citenamefont
  {Gour}(2019)}]{Chitambar.2019.RMP}%
  \BibitemOpen
  \bibfield  {author} {\bibinfo {author} {\bibfnamefont {E.}~\bibnamefont
  {Chitambar}}\ and\ \bibinfo {author} {\bibfnamefont {G.}~\bibnamefont
  {Gour}},\ }\bibfield  {title} {\bibinfo {title} {{Quantum resource
  theories}},\ }\href {https://doi.org/10.1103/RevModPhys.91.025001} {\bibfield
   {journal} {\bibinfo  {journal} {Rev. Mod. Phys.}\ }\textbf {\bibinfo
  {volume} {91}},\ \bibinfo {pages} {025001} (\bibinfo {year}
  {2019})}\BibitemShut {NoStop}%
\bibitem [{\citenamefont {Lipka-Bartosik}\ \emph {et~al.}(2024)\citenamefont
  {Lipka-Bartosik}, \citenamefont {Wilming},\ and\ \citenamefont
  {Ng}}]{Bartosik.2024.RMP}%
  \BibitemOpen
  \bibfield  {author} {\bibinfo {author} {\bibfnamefont {P.}~\bibnamefont
  {Lipka-Bartosik}}, \bibinfo {author} {\bibfnamefont {H.}~\bibnamefont
  {Wilming}},\ and\ \bibinfo {author} {\bibfnamefont {N.~H.~Y.}\ \bibnamefont
  {Ng}},\ }\bibfield  {title} {\bibinfo {title} {{Catalysis in quantum
  information theory}},\ }\href {https://doi.org/10.1103/RevModPhys.96.025005}
  {\bibfield  {journal} {\bibinfo  {journal} {Rev. Mod. Phys.}\ }\textbf
  {\bibinfo {volume} {96}},\ \bibinfo {pages} {025005} (\bibinfo {year}
  {2024})}\BibitemShut {NoStop}%
\bibitem [{\citenamefont {Brandão}\ \emph {et~al.}(2015)\citenamefont
  {Brandão}, \citenamefont {Horodecki}, \citenamefont {Ng}, \citenamefont
  {Oppenheim},\ and\ \citenamefont {Wehner}}]{Brando.2015.PNAS}%
  \BibitemOpen
  \bibfield  {author} {\bibinfo {author} {\bibfnamefont {F.}~\bibnamefont
  {Brandão}}, \bibinfo {author} {\bibfnamefont {M.}~\bibnamefont {Horodecki}},
  \bibinfo {author} {\bibfnamefont {N.}~\bibnamefont {Ng}}, \bibinfo {author}
  {\bibfnamefont {J.}~\bibnamefont {Oppenheim}},\ and\ \bibinfo {author}
  {\bibfnamefont {S.}~\bibnamefont {Wehner}},\ }\bibfield  {title} {\bibinfo
  {title} {{The second laws of quantum thermodynamics}},\ }\href
  {https://doi.org/10.1073/pnas.1411728112} {\bibfield  {journal} {\bibinfo
  {journal} {Proc. Natl. Acad. Sci. U.S.A.}\ }\textbf {\bibinfo {volume}
  {112}},\ \bibinfo {pages} {3275} (\bibinfo {year} {2015})}\BibitemShut
  {NoStop}%
\bibitem [{\citenamefont {Lostaglio}\ \emph {et~al.}(2015)\citenamefont
  {Lostaglio}, \citenamefont {Jennings},\ and\ \citenamefont
  {Rudolph}}]{Lostaglio.2015.NC}%
  \BibitemOpen
  \bibfield  {author} {\bibinfo {author} {\bibfnamefont {M.}~\bibnamefont
  {Lostaglio}}, \bibinfo {author} {\bibfnamefont {D.}~\bibnamefont
  {Jennings}},\ and\ \bibinfo {author} {\bibfnamefont {T.}~\bibnamefont
  {Rudolph}},\ }\bibfield  {title} {\bibinfo {title} {{Description of quantum
  coherence in thermodynamic processes requires constraints beyond free
  energy}},\ }\href {https://doi.org/10.1038/ncomms7383} {\bibfield  {journal}
  {\bibinfo  {journal} {Nat. Commun.}\ }\textbf {\bibinfo {volume} {6}},\
  \bibinfo {pages} {6383} (\bibinfo {year} {2015})}\BibitemShut {NoStop}%
\bibitem [{fnt({\natexlab{a}})}]{fnt1}%
  \BibitemOpen
  \href@noop {} {}While some studies compare relaxation from an initial
  equilibrium state with that from a non-equilibrium state, this Letter focuses
  exclusively on relaxation processes starting from equilibrium
  states.\BibitemShut {Stop}%
\bibitem [{\citenamefont {Schnakenberg}(1976)}]{Schnakenberg.1976.RMP}%
  \BibitemOpen
  \bibfield  {author} {\bibinfo {author} {\bibfnamefont {J.}~\bibnamefont
  {Schnakenberg}},\ }\bibfield  {title} {\bibinfo {title} {{Network theory of
  microscopic and macroscopic behavior of master equation systems}},\ }\href
  {https://doi.org/10.1103/RevModPhys.48.571} {\bibfield  {journal} {\bibinfo
  {journal} {Rev. Mod. Phys.}\ }\textbf {\bibinfo {volume} {48}},\ \bibinfo
  {pages} {571} (\bibinfo {year} {1976})}\BibitemShut {NoStop}%
\bibitem [{\citenamefont {Sagawa}(2022)}]{Sagawa.2022}%
  \BibitemOpen
  \bibfield  {author} {\bibinfo {author} {\bibfnamefont {T.}~\bibnamefont
  {Sagawa}},\ }\href@noop {} {\emph {\bibinfo {title} {{Entropy, divergence,
  and majorization in classical and quantum thermodynamics}}}},\ Vol.~\bibinfo
  {volume} {16}\ (\bibinfo  {publisher} {Springer Nature},\ \bibinfo {year}
  {2022})\BibitemShut {NoStop}%
\bibitem [{\citenamefont {Blackwell}(1953)}]{Blackwell.1953.AMS}%
  \BibitemOpen
  \bibfield  {author} {\bibinfo {author} {\bibfnamefont {D.}~\bibnamefont
  {Blackwell}},\ }\bibfield  {title} {\bibinfo {title} {Equivalent comparisons
  of experiments},\ }\href {https://doi.org/10.1214/aoms/1177729032} {\bibfield
   {journal} {\bibinfo  {journal} {Ann. Math. Stat.}\ }\textbf {\bibinfo
  {volume} {24}},\ \bibinfo {pages} {265} (\bibinfo {year} {1953})}\BibitemShut
  {NoStop}%
\bibitem [{\citenamefont {vom Ende}\ and\ \citenamefont
  {Dirr}(2022)}]{vomEnde.2022.LAA}%
  \BibitemOpen
  \bibfield  {author} {\bibinfo {author} {\bibfnamefont {F.}~\bibnamefont {vom
  Ende}}\ and\ \bibinfo {author} {\bibfnamefont {G.}~\bibnamefont {Dirr}},\
  }\bibfield  {title} {\bibinfo {title} {{The $d$-majorization polytope}},\
  }\href {https://doi.org/10.1016/j.laa.2022.05.005} {\bibfield  {journal}
  {\bibinfo  {journal} {Linear Algebra Appl.}\ }\textbf {\bibinfo {volume}
  {649}},\ \bibinfo {pages} {152} (\bibinfo {year} {2022})}\BibitemShut
  {NoStop}%
\bibitem [{Sup()}]{Supp.PhysRev}%
  \BibitemOpen
  \href@noop {} {}\bibinfo {note} {See the Supplemental Material for the detailed
  proof of the main results, the analysis of the origin of the crossover time
  $t_*$, and the extension of the thermomajorization theory to continuous-state
  cases.}\BibitemShut {Stop}%
\bibitem [{fnt({\natexlab{b}})}]{fnt2}%
  \BibitemOpen
  \href@noop {} {}Since Definition \ref{def:TME} and Result \ref{res:equiv.TME}
  do not require any details of the underlying dynamics, they are broadly
  applicable to any stochastic dynamics. Markovian dynamics is introduced
  specifically to explain Results \ref{res:long.time.tME.cond},
  \ref{res:eig.TME}, and \ref{res:uni.TME}.\BibitemShut {Stop}%
\bibitem [{\citenamefont {Datta}\ \emph {et~al.}(2022)\citenamefont {Datta},
  \citenamefont {Pietzonka},\ and\ \citenamefont {Barato}}]{Datta.2022.PRX}%
  \BibitemOpen
  \bibfield  {author} {\bibinfo {author} {\bibfnamefont {A.}~\bibnamefont
  {Datta}}, \bibinfo {author} {\bibfnamefont {P.}~\bibnamefont {Pietzonka}},\
  and\ \bibinfo {author} {\bibfnamefont {A.~C.}\ \bibnamefont {Barato}},\
  }\bibfield  {title} {\bibinfo {title} {{Second law for active heat
  engines}},\ }\href {https://doi.org/10.1103/PhysRevX.12.031034} {\bibfield
  {journal} {\bibinfo  {journal} {Phys. Rev. X}\ }\textbf {\bibinfo {volume}
  {12}},\ \bibinfo {pages} {031034} (\bibinfo {year} {2022})}\BibitemShut
  {NoStop}%
\bibitem [{\citenamefont {Biswas}\ and\ \citenamefont
  {Pal}(2024)}]{Biswas.2024.arxiv}%
  \BibitemOpen
  \bibfield  {author} {\bibinfo {author} {\bibfnamefont {A.}~\bibnamefont
  {Biswas}}\ and\ \bibinfo {author} {\bibfnamefont {A.}~\bibnamefont {Pal}},\
  }\bibfield  {title} {\bibinfo {title} {{Mpemba effect on non-equilibrium
  active Markov chains}},\ }\href {https://arxiv.org/abs/2403.17547} {\bibfield
   {journal} {\bibinfo  {journal} {arXiv:2403.17547}\ } (\bibinfo {year}
  {2024})}\BibitemShut {NoStop}%
\bibitem [{fnt({\natexlab{c}})}]{fnt3}%
  \BibitemOpen
  \href@noop {} {}We directly determine the crossover time $t_*$ as the minimum
  time satisfying $\vb*{p}_t^h\prec_{\vb*{\pi}}\vb*{p}_t^w$ by verifying
  condition (c), without relying on any $f$-divergences.\BibitemShut {Stop}%
\bibitem [{fnt({\natexlab{d}})}]{fnt4}%
  \BibitemOpen
  \href@noop {} {}The detailed strategy to construct such an instance is
  outlined in Sec.~S4 and Lemma 1 of Ref.~\cite{Supp.PhysRev}.\BibitemShut
  {Stop}%
\bibitem [{\citenamefont {Ohga}\ \emph {et~al.}(2024)\citenamefont {Ohga},
  \citenamefont {Hayakawa},\ and\ \citenamefont {Ito}}]{Ohga.2024.arxiv}%
  \BibitemOpen
  \bibfield  {author} {\bibinfo {author} {\bibfnamefont {N.}~\bibnamefont
  {Ohga}}, \bibinfo {author} {\bibfnamefont {H.}~\bibnamefont {Hayakawa}},\
  and\ \bibinfo {author} {\bibfnamefont {S.}~\bibnamefont {Ito}},\ }\bibfield
  {title} {\bibinfo {title} {{Microscopic theory of Mpemba effects and a
  no-Mpemba theorem for monotone many-body systems}},\ }\href
  {https://arxiv.org/abs/2410.06623} {\bibfield  {journal} {\bibinfo  {journal}
  {arXiv:2410.06623}\ } (\bibinfo {year} {2024})}\BibitemShut {NoStop}%
\bibitem [{\citenamefont {Deffner}\ and\ \citenamefont
  {Campbell}(2017)}]{Deffner.2017.JPA}%
  \BibitemOpen
  \bibfield  {author} {\bibinfo {author} {\bibfnamefont {S.}~\bibnamefont
  {Deffner}}\ and\ \bibinfo {author} {\bibfnamefont {S.}~\bibnamefont
  {Campbell}},\ }\bibfield  {title} {\bibinfo {title} {{Quantum speed limits:
  from Heisenberg's uncertainty principle to optimal quantum control}},\ }\href
  {https://doi.org/10.1088/1751-8121/aa86c6} {\bibfield  {journal} {\bibinfo
  {journal} {J. Phys. A}\ }\textbf {\bibinfo {volume} {50}},\ \bibinfo {pages}
  {453001} (\bibinfo {year} {2017})}\BibitemShut {NoStop}%
\bibitem [{\citenamefont {Van~Vu}\ and\ \citenamefont
  {Saito}(2023)}]{Vu.2023.PRL.TSL}%
  \BibitemOpen
  \bibfield  {author} {\bibinfo {author} {\bibfnamefont {T.}~\bibnamefont
  {Van~Vu}}\ and\ \bibinfo {author} {\bibfnamefont {K.}~\bibnamefont {Saito}},\
  }\bibfield  {title} {\bibinfo {title} {{Topological speed limit}},\ }\href
  {https://doi.org/10.1103/PhysRevLett.130.010402} {\bibfield  {journal}
  {\bibinfo  {journal} {Phys. Rev. Lett.}\ }\textbf {\bibinfo {volume} {130}},\
  \bibinfo {pages} {010402} (\bibinfo {year} {2023})}\BibitemShut {NoStop}%
\bibitem [{\citenamefont {Srivastav}\ \emph {et~al.}(2024)\citenamefont
  {Srivastav}, \citenamefont {Pandey}, \citenamefont {Mohan},\ and\
  \citenamefont {Pati}}]{Srivastav.2024.arxiv}%
  \BibitemOpen
  \bibfield  {author} {\bibinfo {author} {\bibfnamefont {A.}~\bibnamefont
  {Srivastav}}, \bibinfo {author} {\bibfnamefont {V.}~\bibnamefont {Pandey}},
  \bibinfo {author} {\bibfnamefont {B.}~\bibnamefont {Mohan}},\ and\ \bibinfo
  {author} {\bibfnamefont {A.~K.}\ \bibnamefont {Pati}},\ }\bibfield  {title}
  {\bibinfo {title} {{Family of exact and inexact quantum speed limits for
  completely positive and trace-preserving dynamics}},\ }\href
  {https://arxiv.org/abs/2406.08584} {\bibfield  {journal} {\bibinfo  {journal}
  {arXiv:2406.08584}\ } (\bibinfo {year} {2024})}\BibitemShut {NoStop}%
\end{thebibliography}
\end{document}

\title{Supplemental Material for ``Thermomajorization Mpemba Effect''}

\author{Tan Van Vu}
\email{tan.vu@yukawa.kyoto-u.ac.jp}
\affiliation{Center for Gravitational Physics and Quantum Information, Yukawa Institute for Theoretical Physics, Kyoto University, Kitashirakawa Oiwakecho, Sakyo-ku, Kyoto 606-8502, Japan}

\author{Hisao Hayakawa}
\affiliation{Center for Gravitational Physics and Quantum Information, Yukawa Institute for Theoretical Physics, Kyoto University, Kitashirakawa Oiwakecho, Sakyo-ku, Kyoto 606-8502, Japan}

\begin{abstract}
This Supplemental Material describes the proof of the analytical results obtained in the main text, the analysis of the origin of the crossover time $t_*$, and the extension of the thermomajorization theory to continuous-state cases. 
The equations and figure numbers are prefixed with S [e.g., Eq.~(S1) or Fig.~S1]. 
The numbers without this prefix [e.g., Eq.~(1) or Fig.~1] refer to the items in the main text.
\end{abstract}

\pacs{}
\maketitle

\tableofcontents

\section{Proof of the thermomajorization order $\vpi^w\prec_{\vpi}\vpi^h$}
Here we prove that $\vpi^h$ thermomajorizes $\vpi^w$ as long as $T<T_w<T_h$.
Let $Z$ and $Z_s$ denote the partition function of the Gibbs states $\vpi$ and $\vpi^s$, respectively.
Since $\pi_n^s/\pi_n=e^{(\beta-\beta_s)\epsilon_n}ZZ_s^{-1}$, $0\le \epsilon_1\le\dots\le \epsilon_d$, and $\beta-\beta_s\ge 0$ for $s\in\{h,w\}$, the rearranged distributions $\tilde{\vpi}^s$ and $\tilde{\vpi}$ can be easily determined as $\tilde{\vpi}^s=[\pi_d^s,\dots,\pi_1^s]^\top$ and $\tilde{\vpi}=[\pi_d,\dots,\pi_1]^\top$.
Therefore, the Lorenz curve of $(\vpi^s,\vpi)$ is given by a polyline $OP_1^sP_2^s\dots P_d^s$, where $O=(0,0)$ and
\begin{equation}
	P_n^s\coloneqq(x_n^s,y_n^s)=\qty(\sum_{k=d+1-n}^{d}\pi_{k},\sum_{k=d+1-n}^{d}\pi_k^s)~\forall n=1,\dots,d.
\end{equation}
Since $x_n^w=x_n^h$, in order to prove $\vpi^w\prec_{\vpi}\vpi^h$, it is sufficient to show that the following inequality holds true for any $n$:
\begin{equation}
	y_n^w\le y_n^h.\label{eq:mar.tmp1}
\end{equation}
To this end, we first note that $\pi_1^s\ge\pi_2^s\ge\dots\ge\pi_d^s$ and
\begin{align}
	\pi_1^h&=\frac{1}{\sum_{n=1}^de^{\beta_h(\epsilon_1-\epsilon_n)}}\le\frac{1}{\sum_{n=1}^de^{\beta_w(\epsilon_1-\epsilon_n)}}=\pi_1^w,\\
	\pi_d^h&=\frac{1}{\sum_{n=1}^de^{\beta_h(\epsilon_d-\epsilon_n)}}\ge\frac{1}{\sum_{n=1}^de^{\beta_w(\epsilon_d-\epsilon_n)}}=\pi_d^w.
\end{align}
In addition, we have
\begin{align}
	\frac{\pi_k^h}{\pi_{k+1}^h}\frac{\pi_{k+1}^w}{\pi_k^w}=e^{(\beta_h-\beta_w)(\epsilon_{k+1}-\epsilon_k)}\le 1\rightarrow\frac{\pi_k^h}{\pi_k^w}\le\frac{\pi_{k+1}^h}{\pi_{k+1}^w}.
\end{align}
Therefore, there exists an index $i$ such that
\begin{equation}
	\frac{\pi_1^h}{\pi_1^w}\le\dots\le\frac{\pi_i^h}{\pi_i^w}\le 1\le\frac{\pi_{i+1}^h}{\pi_{i+1}^w}\le\dots\le\frac{\pi_{d}^h}{\pi_{d}^w},
\end{equation}
or equivalently $\pi_k^h\ge\pi_k^w$ for any $k\ge i+1$ and $\pi_k^h\le\pi_k^w$ for any $k\le i$.
Now, inequality \eqref{eq:mar.tmp1} can be proved as follows.
For $n\le d-i$, it can be easily verified by noticing that
\begin{align}
	y_n^h-y_n^w=\sum_{k=d+1-n}^d(\pi_k^h-\pi_k^w)\ge 0.
\end{align}
For $n\ge d-i+1$, since $\sum_{k=i+1}^d(\pi_k^h-\pi_k^w)=\sum_{k=1}^i(\pi_k^w-\pi_k^h)$, we can calculate as
\begin{align}
	y_n^h-y_n^w&=\sum_{k=i+1}^d(\pi_k^h-\pi_k^w)+\sum_{k=d+1-n}^i(\pi_k^h-\pi_k^w)\notag\\
	&=\sum_{k=1}^i(\pi_k^w-\pi_k^h)+\sum_{k=d+1-n}^i(\pi_k^h-\pi_k^w)\notag\\
	&=\sum_{k=1}^{d-n}(\pi_k^w-\pi_k^h)\ge 0,
\end{align}
which validates inequality \eqref{eq:mar.tmp1}.

\section{Proof of Result \FirRes}
We first prove that the thermomajorization Mpemba effect implies the Mpemba effect within a finite bounded time for \emph{all} monotone measures.
To this end, let $D$ be an arbitrary monotone measure; that is, $D(\vecp,\vecq)\ge D(\msf{T}\vecp,\msf{T}\vecq)$ for any stochastic matrix $\msf{T}$.
Since $\vpi^w\prec_{\vpi}\vpi^h$, there exists a stochastic matrix $\msf{T}_0$ such that $\msf{T}_0\vpi^h=\vpi^w$ and $\msf{T}_0\vpi=\vpi$, according to condition (a).
Due to the monotonicity of measure $D$, we obtain
\begin{equation}
	D(\vpi^h,\vpi)\ge D(\msf{T}_0\vpi^h,\msf{T}_0\vpi)=D(\vpi^w,\vpi),
\end{equation}
which indicates that the warm state is initially closer to the final thermal state than the hot one when quantified by measure $D$.
It is thus sufficient to prove $D(\vecp_t^h,\vpi)\le D(\vecp_t^w,\vpi)$, given that the thermomajorization Mpemba effect occurs at time $t<\infty$.
Since $\vpt^h\prec_{\vpi}\vpt^w$, there exists a stochastic matrix $\msf{T}_t$ such that $\msf{T}_t\vpt^w=\vpt^h$ and $\msf{T}_t\vpi=\vpi$.
By applying the monotonicity of the measure $D$ again, we achieve
\begin{equation}
	D(\vpt^w,\vpi)\ge D(\msf{T}_t\vpt^w,\msf{T}_t\vpi)=D(\vpt^h,\vpi).
\end{equation}
This implies that there exists a crossover within the bounded time $t$ between the relaxation trends from the initial hot and warm thermal states, which is nothing but the conventional Mpemba effect based on measure $D$.

Next, it remains to prove the converse statement.
Assume that the conventional Mpemba effect occurs for all monotone measures within a finite bounded time $t$; that is, $D(\vpi^h,\vpi)>D(\vpi^w,\vpi)$ and $D(\vecp_t^h,\vpi)\le D(\vecp_t^w,\vpi)$ are simultaneously satisfied for any monotone measure $D$.
We need to show the existence of the thermomajorization Mpemba effect in this case.
Let $f(x)$ be an arbitrary continuous, convex function and define $\bar{f}(x)\coloneqq f(x)-f(1)$.
Then, $\bar{f}(1)=0$ and $\bar{f}(x)$ is also a continuous, convex function.
A monotone measure based on this convex function can be defined as
\begin{equation}
	D_{\bar{f}}(\vecp,\vecq)\coloneqq\sum_nq_n\bar{f}(p_n/q_n)\ge 0,
\end{equation}
which is also known as the $\bar{f}$-divergence.
Since the conventional Mpemba effect based on the monotone measure $D_{\bar{f}}$ occurs within time $t$, the crossover condition $D_{\bar{f}}(\vpt^h,\vpi)\le D_{\bar{f}}(\vpt^w,\vpi)$ immediately derives
\begin{equation}
	\sum_{n=1}^d\pi_nf\qty[\frac{p_n^h(t)}{\pi_n}]\le\sum_{n=1}^d\pi_nf\qty[\frac{p_n^w(t)}{\pi_n}],
\end{equation}
which holds true for arbitrary continuous, convex functions $f(x)$.
According to condition (b), this implies $\vpt^h\prec_{\vpi}\vpt^w$, which is exactly the thermomajorization Mpemba effect.

\section{Proof of Result \SecRes}
In the long-time regime, the probability distribution can be expressed as
\begin{equation}\label{eq:prob.long.time}
	\vpt^s=\vpi+e^{\lambda_2t}\vecv^s+O(e^{\lambda_{\kappa+1}t}),
\end{equation}
where $\vecv^s=\sum_{n=2}^\kappa a_2^s\vecr_n$.
According to condition (c) in the main text, the following inequality must hold true in order to achieve the thermomajorization Mpemba effect:
\begin{equation}\label{eq:tME.condc}
	\|\vpt^h-z\vpi\|_1\le\|\vpt^w-z\vpi\|_1~\forall z\in\mbb{R}.
\end{equation}
Inequality \eqref{eq:tME.condc} is sufficiently guaranteed by examining $z\in\{p_1^w(t)/\pi_1,\dots,p_d^w(t)/\pi_d,p_1^h(t)/\pi_1,\dots,p_d^h(t)/\pi_d\}$ \cite{Marshall.2011}.
Inserting these values of $z$ into inequality \eqref{eq:tME.condc} and using the expression \eqref{eq:prob.long.time}, we obtain the following criterion:
\begin{align}
	\|\vecv^h-z\vpi\|_1\le\|\vecv^w-z\vpi\|_1~\forall z\in\qty{\frac{v_k^w}{\pi_k},\frac{v_k^h}{\pi_k}}_{1\le k\le d}.\label{eq:tME.long.cond}
\end{align}
Furthermore, since the Lorenz curve of $(\vpt^w,\vpi)$ must lie strictly above that of $(\vpt^h,\vpi)$, at least one of inequalities \eqref{eq:tME.long.cond} must hold strictly.
In the nondegenerate case (i.e., $\lambda_2>\lambda_3$), these inequalities can be simplified to
\begin{align}
	\Big\|\frac{a_2^h}{a_2^w}\vecr_2-z\vpi\Big\|_1\le\|\vecr_2-z\vpi\|_1\le\Big\|\frac{a_2^w}{a_2^h}\vecr_2-z\vpi\Big\|_1~\forall z\in\qty{\frac{r_{2k}}{\pi_k}}_{1\le k\le d}.\label{eq:tME.long.cond}
\end{align}

\section{Proof of Result \ThiRes}
We prove the first statement that the thermomajorization Mpemba effect in the long-time regime derives $|a_2^h|<|a_2^w|$.
It is sufficient to prove that inequalities (\LongTimeNonDegCond) with at least one of these inequalities holding strictly,
\begin{align}
	\sum_{n=1}^d\qty|r_{2n}-\frac{\pi_n}{\pi_k}r_{2k}|&\ge\sum_{n=1}^d\qty|\frac{a_2^h}{a_2^w}r_{2n}-\frac{\pi_n}{\pi_k}r_{2k}|,\label{eq:nondeg.cond1}\\
	\sum_{n=1}^d\qty|r_{2n}-\frac{\pi_n}{\pi_k}r_{2k}|&\le\sum_{n=1}^d\qty|\frac{a_2^w}{a_2^h}r_{2n}-\frac{\pi_n}{\pi_k}r_{2k}|,\label{eq:nondeg.cond2}
\end{align}
will imply $|a_2^h|<|a_2^w|$.

To this end, we assume $|a_2^h|\ge|a_2^w|$ and show that this leads to a contradiction.
Let $i=\argmin_kr_{2k}/\pi_k$ and $j=\argmax_kr_{2k}/\pi_k$.
The conditions $\vb*{1}^\top\vecr_2=0$ and $\vecr_2\neq\vb*{0}$ readily imply that $r_{2i}<0$ and $r_{2j}>0$.
Applying the triangle inequality, we get
\begin{align}
	\sum_{n=1}^d\qty|\frac{a_2^h}{a_2^w}r_{2n}-\frac{\pi_n}{\pi_i}r_{2i}|&\ge\qty|\sum_{n=1}^d\qty(\frac{a_2^h}{a_2^w}r_{2n}-\frac{\pi_n}{\pi_i}r_{2i})|\notag\\
	&=\qty|\frac{r_{2i}}{\pi_i}|\notag\\
	&=\qty|\sum_{n=1}^d\qty(r_{2n}-\frac{\pi_n}{\pi_i}r_{2i})|\notag\\
	&=\sum_{n=1}^d\qty|r_{2n}-\frac{\pi_n}{\pi_i}r_{2i}|.\label{eq:nondeg.cond.tmp1}
\end{align}
Combining inequalities \eqref{eq:nondeg.cond1} and \eqref{eq:nondeg.cond.tmp1} results in their equality.
In other words, the sign of $(a_2^h/a_2^w)r_{2n}-(\pi_n/\pi_i)r_{2i}$ must be the same for all $n\in[1,d]$.
Similarly, we obtain
\begin{align}
	\sum_{n=1}^d\qty|\frac{a_2^h}{a_2^w}r_{2n}-\frac{\pi_n}{\pi_j}r_{2j}|&\ge\qty|\sum_{n=1}^d\qty(\frac{a_2^h}{a_2^w}r_{2n}-\frac{\pi_n}{\pi_j}r_{2j})|\notag\\
	&=\qty|\frac{r_{2j}}{\pi_j}|\notag\\
	&=\qty|\sum_{n=1}^d\qty(r_{2n}-\frac{\pi_n}{\pi_j}r_{2j})|\notag\\
	&=\sum_{n=1}^d\qty|r_{2n}-\frac{\pi_n}{\pi_j}r_{2j}|,\label{eq:nondeg.cond.tmp2}
\end{align}
which indicates that the sign of $(a_2^h/a_2^w)r_{2n}-(\pi_n/\pi_j)r_{2j}$ must also be the same for all $n\in[1,d]$.

Evidently, $|a_2^h|\neq|a_2^w|$ because if $a_2^h/a_2^w=\pm 1$, all inequalities \eqref{eq:nondeg.cond1} and \eqref{eq:nondeg.cond2} will become equalities, violating the condition that at least one of these inequalities must hold strictly.
Consider two cases: $a_2^h/a_2^w>1$ and $a_2^h/a_2^w<-1$. 
If $a_2^h/a_2^w>1$, the following terms have opposite signs:
\begin{align}
	\frac{a_2^h}{a_2^w}r_{2i}-\frac{\pi_i}{\pi_i}r_{2i}&=\qty(\frac{a_2^h}{a_2^w}-1)r_{2i}<0,\\
	\frac{a_2^h}{a_2^w}r_{2j}-\frac{\pi_j}{\pi_i}r_{2i}&>0.
\end{align}
This contradicts the sign condition derived above.
If $a_2^h/a_2^w<-1$, since $(a_2^h/a_2^w)r_{2i}-(\pi_i/\pi_i)r_{2i}>0$ and $(a_2^h/a_2^w)r_{2j}-(\pi_j/\pi_j)r_{2j}<0$, we obtain the following inequalities from the sign conditions:
\begin{align}
	\frac{a_2^h}{a_2^w}r_{2j}-\frac{\pi_j}{\pi_i}r_{2i}&\ge 0,\label{eq:nondeg.cond.tmp3}\\
	\frac{a_2^h}{a_2^w}r_{2i}-\frac{\pi_i}{\pi_j}r_{2j}&\le 0.\label{eq:nondeg.cond.tmp4}
\end{align}
Noting that $r_{2i}<0$ and $r_{2j}>0$, we further achieve
\begin{equation}
	-\frac{r_{2i}}{\pi_i}\ge\qty(-\frac{a_2^h}{a_2^w})\frac{r_{2j}}{\pi_j}\ge \qty(-\frac{a_2^h}{a_2^w})^2\qty(-\frac{r_{2i}}{\pi_i})>-\frac{r_{2i}}{\pi_i},
\end{equation}
which is a contradiction.
Therefore, $|a_2^h|<|a_2^w|$ must hold true.

Next, we prove that this simple relationship between the coefficient of the slowest decay mode does not imply the thermomajorization Mpemba effect at a long time.
To this end, it is sufficient to show an instance of a transition matrix that fulfills $|a_2^h|<|a_2^w|$ but violates inequality \eqref{eq:nondeg.cond1} for some $k$.
It is a key observation that for any vector $\vecr_2$ orthogonal to $\vb*{1}$, one can always find a set of barrier coefficients such that the resulting transition matrix is nondegenerate and takes $\vecr_2$ as the right eigenvector corresponding to the second largest eigenvalue (cf.~Lemma \ref{lem:tran.mat.exist}).
Since $\vb*{1}^\top\vecr_2=0$, we have $\vpi^\top\vecx=0$, where $\vecx\coloneqq\Pi^{-1}\vecr_2$ and $\Pi\coloneqq\diag(\pi_1,\dots,\pi_d)$ is a diagonal matrix.
Additionally, by noting that $a_2^s=(\vecx^\top\vpi^s)/(\vecx^\top\vecr_2)$ as $\vecx^\top\vecr_n=0$ for any $n\neq 2$, inequality \eqref{eq:nondeg.cond1} can be expressed in the following form:
\begin{equation}\label{eq:nondeg.tmp}
	\sum_{n=1}^d\pi_n|x_n-x_k|\ge\sum_{n=1}^d\pi_n\qty|\frac{\vecx^\top\vpi^h}{\vecx^\top\vpi^w}x_n-x_k|.
\end{equation}
Therefore, we need only show a violation of inequality \eqref{eq:nondeg.tmp} under the constraints $\vpi^\top\vecx=0$ and $|\vecx^\top\vpi^w|>|\vecx^\top\vpi^h|$.
This can be easily achieved, for instance, in a three-level system with $T_h=1.3$, $T_w=0.42$, $T=0.1$, $\vb*{\epsilon}=[0,0.1,0.7]^\top$, $x_1=0.23$, $x_3=0.73$, and $k=3$.
A detailed construction of the transition matrix that achieves the vector $\vecx$ is provided in Lemma \ref{lem:tran.mat.exist}.

\section{Proof of Result \FouRes}
We prove the result by analytically constructing a transition matrix that satisfies the detailed balance condition with respect to the fixed energy levels $\{\epsilon_n\}$ and displays the thermomajorization Mpemba effect.
To this end, let $\{\vecv_1,\vecv_2\}$ be an orthonormal basis of the space $\mca{V}$ spanned by $\vpi^h$ and $\vpi$, and $\msf{P}\coloneqq\vecv_1\vecv_1^\top+\vecv_2\vecv_2^\top$ be the projection matrix onto the space $\mca{V}$.
Note that $\msf{P}\vpi^h=\vpi^h$ and $\msf{P}\vpi=\vpi$.
We define $\Pi\coloneqq\diag(\vpi)$ and $\vecr_2\coloneqq\Pi(\vpi^w-\msf{P}\vpi^w)$, which fulfills $\vecr_2\neq\vb*{0}$ because the vectors $\vpi^w,\vpi^h$, and $\vpi$ are linearly independent \cite{Vu.2021.PRR}.
Furthermore, it can be verified that $\vb*{1}^\top\vecr_2=0$.
Consequently, one can find a configuration of barrier coefficients such that the transition matrix is nondegenerate and takes $\vecr_2$ as the eigenvector of the second largest eigenvalue (cf.~Lemma \ref{lem:tran.mat.exist}).
We need only prove that this transition matrix exhibits the thermomajorization Mpemba effect.
In other words, it is sufficient to validate inequalities (\LongTimeNonDegCond).
Note that the coefficient associated with the slowest decay mode can be calculated as
\begin{equation}
	a_2^s=\frac{\vecr_2^\top\Pi^{-1}\vpi^s}{\vecr_2^\top\Pi^{-1}\vecr_2}.
\end{equation}
We can confirm that
\begin{equation}
	a_2^h=\frac{\vecr_2^\top\Pi^{-1}\vpi^h}{\vecr_2^\top\Pi^{-1}\vecr_2}=\frac{(\vpi^w-\msf{P}\vpi^w)^\top\vpi^h}{\vecr_2^\top\Pi^{-1}\vecr_2}=0
\end{equation}
and
\begin{equation}
	a_2^w=\frac{\vecr_2^\top\Pi^{-1}\vpi^w}{\vecr_2^\top\Pi^{-1}\vecr_2}=\frac{\vecr_2^\top\Pi^{-2}\vecr_2}{\vecr_2^\top\Pi^{-1}\vecr_2}>0.
\end{equation}
Using these facts and applying the triangle inequality, inequalities (\LongTimeNonDegCond) can be proved as follows:
\begin{align}
	\sum_{n=1}^d\qty|r_{2n}-\frac{\pi_n}{\pi_k}r_{2k}|&\ge\qty|\sum_{n=1}^d\qty(r_{2n}-\frac{\pi_n}{\pi_k}r_{2k})|\notag\\
	&=\frac{|r_{2k}|}{\pi_k}\notag\\
	&=\sum_{n=1}^d\qty|\frac{a_2^h}{a_2^w}r_{2n}-\frac{\pi_n}{\pi_k}r_{2k}|,\\
	\sum_{n=1}^d\qty|r_{2n}-\frac{\pi_n}{\pi_k}r_{2k}|&<\sum_{n=1}^d\qty|\frac{a_2^w}{a_2^h}r_{2n}-\frac{\pi_n}{\pi_k}r_{2k}|=\infty.
\end{align}
This indicates that this transition matrix displays the thermomajorization Mpemba effect in the long-time regime.

\section{Detailed analysis on the origin of $t_*\approx 16.87$ in the main text}
Here we provide an in-depth analysis of the origin of $t_*\approx 16.87$ in the example discussed in the main text.
To better visualize the behavior of $f_\alpha$-divergence, we plot the function $f_\alpha(x)$ in Fig.~\ref{fig:AlphaD} for various values of $\alpha$. As can be observed, $f_2(x)$ is a symmetric function about $x=1$, whereas increasing $\alpha$ enhances its asymmetry. Particularly, for large $\alpha$, the function $f_\alpha(x)$ tends to place significant weight on values where $x>1$.

\begin{figure}[t]
\centering
\includegraphics[width=1\linewidth]{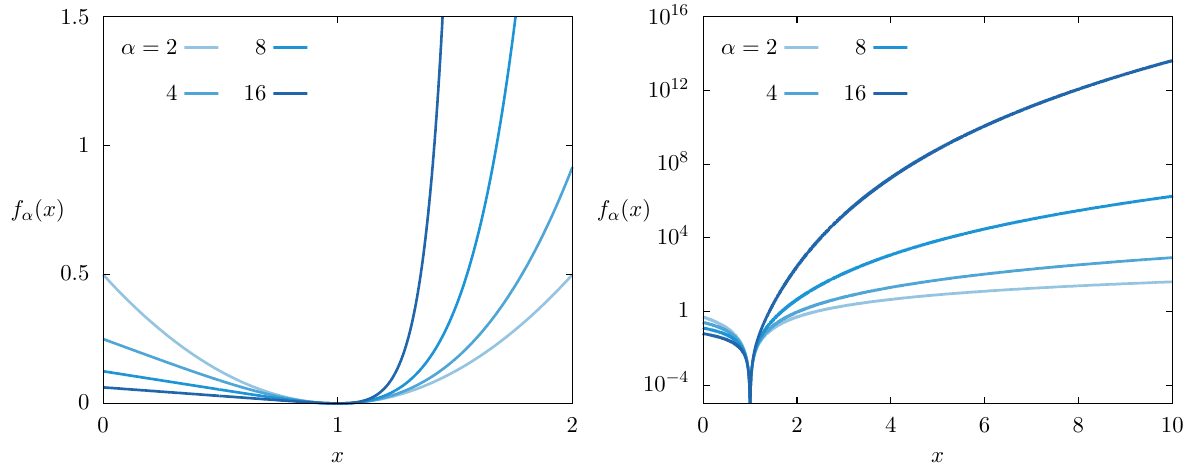}
\protect\caption{Illustration of the behavior of $f_\alpha$-divergence for various values of $\alpha\in\{2,4,8,16\}$. The left panel shows the function $f_\alpha(x)$ on a normal scale, while the right panel presents it on a logarithmic scale.}\label{fig:AlphaD}
\end{figure}

In the $\alpha\to\infty$ limit, we can show that
\begin{equation}
	\lim_{\alpha\to\infty}\frac{1}{\alpha}\ln D_{f_\alpha}(\vb*{p},\vb*{q})=\lim_{\alpha\to\infty}\frac{1}{\alpha}\ln\qty[\sum_{n=1}^dq_nf_\alpha\qty(\frac{p_n}{q_n})]=\max_{1\le n\le d}\ln\frac{p_n}{q_n}.
\end{equation}
This result implies that the $f_\alpha$-divergence between distributions $\vb*{p}$ and $\vb*{q}$ for large $\alpha$ is highly sensitive to the point where $\vb*{p}$ and $\vb*{q}$ differ the most.
In the context of the Mpemba effect, the condition $D_{f_\alpha}(\vb*{p}_{t_*}^h,\vb*{\pi})\le D_{f_\alpha}(\vb*{p}_{t_*}^w,\vb*{\pi})$ for $\alpha\to\infty$ leads to
\begin{equation}\label{eq:tmp.t.org}
	\max_{1\le n\le d}\frac{p_n^h(t_*)}{\pi_n}\le\max_{1\le n\le d}\frac{p_n^w(t_*)}{\pi_n}.
\end{equation}
The crossover time $t_*$ in this case originates from the $f_\infty$-divergence, where the maximum ratio of probabilities must satisfy inequality \eqref{eq:tmp.t.org}.
As shown, the class of $f_\alpha$-divergences not only assesses the overall distance between distributions but also effectively evaluates the largest discrepancies between them. 
This property explains why observing the Mpemba effect using all the monotone measures is both difficult and stringent.

\section{Thermomajorization theory}
\subsection{Properties of the thermomajorization order $\prec_{\vpi}$}
Before introducing the thermomajorization theory in continuous-state cases, we first discuss some fundamental properties of the thermomajorization order $\prec_{\vpi}$.
It is worth noting that $\prec_{\vpi}$ satisfies the axioms of a \emph{partial} order on the space of probability distributions. 
Specifically, it fulfills the following properties:
\begin{enumerate}
	\item Reflexivity: Any probability distribution thermomajorizes itself, $\vb*{p}\prec_{\vpi}\vb*{p}$.	
	\item Antisymmetry: If $\vb*{p}\prec_{\vpi}\vb*{p}'$  and $\vb*{p}'\prec_{\vpi}\vb*{p}$, then the Lorenz curves of $(\vecp,\vpi)$ and $(\vecp',\vpi)$ must be identical.
	
	\item Transitivity: If $\vecp'\prec_{\vpi}\vecp$ and $\vecp''\prec_{\vpi}\vecp'$, then $\vecp''\prec_{\vpi}\vecp$.
\end{enumerate}
We note that the thermomajorization order $\prec_{\vpi}$ is not a \emph{total} order, where every pair of distributions can be compared.
It is only a partial order because not every pair of distributions are comparable under thermomajorization.
That is, for two distributions $\vb*{p}$ and $\vb*{p}'$, it is possible that neither $\vb*{p}\prec_{\vpi}\vb*{p}'$ nor $\vb*{p}'\prec_{\vpi}\vb*{p}$ [i.e., the Lorenz curves of $(\vecp,\vpi)$ and $(\vecp',\vpi)$ strictly intersect without overlapping]. In such cases, the distributions are said to be incomparable under thermomajorization.
Additionally, since $\prec_{\vpi}$ is not a total order, it is clear that no single monotone measure can fully capture the relation defined by $\prec_{\vpi}$.

\subsection{Continuous-state cases}
Here we describe the thermomajorization theory for probability distributions on the continuous-state space \cite{Ruch.1978.JCP,Sagawa.2022}.
Consider continuous probability distributions supported on the domain $\mca{D}\subseteq\mbb{R}$.
For simplicity, we take $\mca{D}=[0,1]$, although the theory straightforwardly generalizes to other domains, including $\mca{D}=\mbb{R}$. 
A valid probability distribution $p:[0,1]\mapsto\mbb{R}_{\ge 0}$ satisfies $p(x)\ge 0~\forall x\in[0,1]$ and $\int_{0}^{1}\dd{x}p(x)=1$.

\subsubsection{Level function and rearrangement}
To establish the thermomajorization theory, it is convenient to define the level function $m_p:[0,+\infty)\mapsto[0,1]$,
\begin{equation}
	m_p(y)\coloneqq\mu\qty{x\in[0,1]|\,p(x)>y},
\end{equation}
where $\mu$ denotes the Lebesgue measure.
The level function $m_p(y)$ is non-increasing, right-continuous, and well-defined for all $y\in[0,+\infty)$.
Intuitively, $m_p(y)$ characterizes the size of the region where the distribution function $p(x)$ exceeds $y$.
It is noteworthy that different distributions may have the same level function, i.e., $m_p=m_q$ for some $p\neq q$.
Using this level function, we can define the rearrangement of $p(x)$ as
\begin{equation}
	p^{\downarrow}(x)\coloneqq\sup\{y\in[0,+\infty)|\,m_p(y)>x\}.
\end{equation}
It can be verified that $p^{\downarrow}$ is a non-increasing function and satisfies
\begin{align}
	\int_0^s\dd{x}p^{\downarrow}(x)&\ge\int_0^s\dd{x}p(x)~\forall s\in[0,1),\notag\\
	\int_0^1\dd{x}p^{\downarrow}(x)&=\int_0^1\dd{x}p(x)=1.
\end{align}

\subsubsection{Thermomajorization order}
With the notion of non-increasing rearrangement established, we are now ready to define the thermomajorization order.
Let $\pi$ be a reference thermal distribution $\pi(x)$, which satisfies $\pi(x)>0~\forall x\in[0,1]$.
Define the cumulative distribution function $F_\pi(x)\coloneqq\int_0^x\dd{x'}\pi(x')$, which derives $\dd{F_\pi(x)}=\pi(x)\dd{x}$, and let $F_\pi^{-1}$ be its inverse function.
For any pair of probability distributions $(p,\pi)$, we define the rescaled distribution function as
\begin{equation}
	p_\pi(x)\coloneqq\frac{p(F_\pi^{-1}(x))}{\pi(F_\pi^{-1}(x))}.
\end{equation}
By the change of variable $x\mapsto F_\pi(x)$, it follows that $\int_0^1\dd{x}p_\pi(x)=\int_0^1\dd{x}p(x)=1$ .
The Lorenz curve of the pair $(p,\pi)$ can be attained from the rearranged distribution $p_\pi^\downarrow(x)$, which represents $p(x)/\pi(x)$ in descending order.
Specifically, the Lorenz curve is given by the concave line
\begin{equation}
	\ell_{p,\pi}\coloneqq\qty{(s,\int_0^s\dd{x}p_\pi^\downarrow(x))}_{0\le s\le 1}.
\end{equation}
We say $p$ thermomajorizes $p'$ with respect to $\pi$, denoted $p'\prec_\pi p$, if the Lorenz curve of $(p,\pi)$ lies above that of $(p',\pi)$.
A geometric illustration of the continuous-state thermomajorization theory is provided in Fig.~\ref{fig:ContThermo}.
\begin{figure}[t]
\centering
\includegraphics[width=1\linewidth]{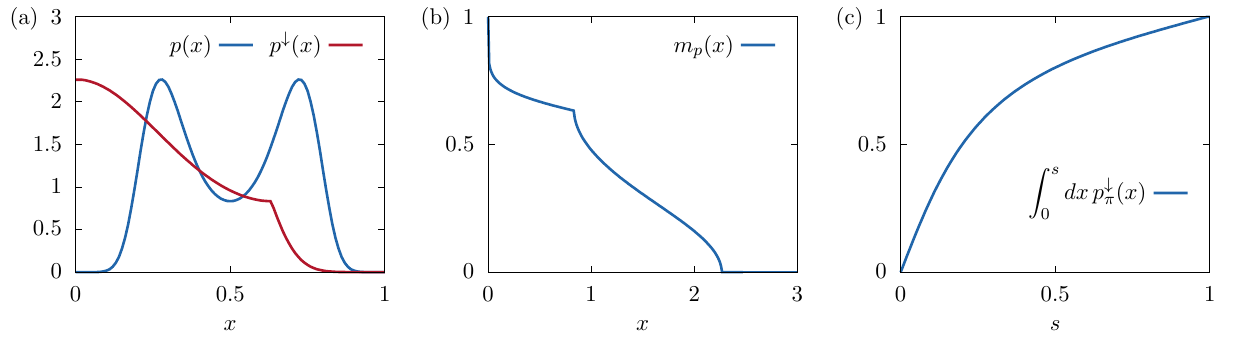}
\protect\caption{Illustration of the continuous-state thermomajorization theory. (a) Probability distribution $p(x)$ and its non-increasing rearrangement $p^\downarrow(x)$. (b) The level function $m_p(x)$. (c) The Lorenz curve of the pair $(p,\pi)$. Here, $p(x)=\exp[-400((x-0.5)^4-0.1(x-0.5)^2)]/Z_p$ and $\pi(x)=\exp[-20(x-0.5)^2]/Z_\pi$ are defined over $x\in[0,1]$, with $Z_p$ and $Z_\pi$ being the normalization constants.}\label{fig:ContThermo}
\end{figure}

\subsubsection{Equivalent conditions}
The thermomajorization order $p'\prec_\pi p$ is equivalent to each of the following conditions:
\begin{itemize}
	\item[($\mathrm{a}'$)] There exists a stochastic map $\mca{T}$ such that $\mca{T}p=p'$ and $\mca{T}\pi=\pi$.
	\item[($\mathrm{b}'$)] $\displaystyle\int_0^1\dd{x}\pi(x)f\qty[\dfrac{p'(x)}{\pi(x)}]\le\displaystyle\int_0^1\dd{x}\pi(x)f\qty[\dfrac{p(x)}{\pi(x)}]$ for any continuous, convex function $f$.
	\item[($\mathrm{c}'$)] $\displaystyle\int_0^1\dd{x}|p'(x)-z\pi(x)|\le\displaystyle\int_0^1\dd{x}|p(x)-z\pi(x)|~\forall z\in\mbb{R}$.
\end{itemize}
Here, a continuous stochastic map $\mca{T}$ is defined as
\begin{equation}
	(\mca{T}p)(x')\coloneqq\frac{d}{dx'}\int_0^1\dd{x}K(x',x)p(x),
\end{equation}
where $K(x',x)$ is a monotonically increasing function of $x'$ satisfying $K(0,x)=0$, $K(1,x)=1$, and $\int_0^1\dd{x}K(x',x)=x'$.

\section{Lemma on the existence of a transition matrix with a given eigenvector}
\begin{lemma}\label{lem:tran.mat.exist}
For any fixed inverse temperature $\beta$, energy levels $\{\epsilon_n\}$, and nonzero vector $\vecr_2$ orthogonal to $\vb*{1}$, there exists a configuration of barrier coefficients $\{b_{mn}\}$ such that the resulting transition matrix $\msf{W}$ satisfies the following conditions:
\begin{enumerate}
	\item $\msf{W}$ is nondegenerate, i.e., its eigenvalues are distinct.
	\item $\msf{W}$ fulfills the detailed balance condition, i.e., $w_{mn}e^{-\beta \epsilon_n}=w_{nm}e^{-\beta \epsilon_m}~\forall m,n$.
	\item $\vecr_2$ is the right eigenvector corresponding to the second largest eigenvalue of $\msf{W}$.
\end{enumerate}
\end{lemma}
\begin{proof}
We define a reference transition matrix $\msf{W}^{\rm ref}$, determined by setting $b_{mn}^{\rm ref}=\epsilon_m+\epsilon_n$.
It can be easily confirmed that the matrix $\msf{X}^{\rm ref}=\Pi^{-1/2}\msf{W}^{\rm ref}\Pi^{1/2}$ is symmetric; additionally, $\msf{X}^{\rm ref}$ has a single zero eigenvalue associated with the eigenvector $\vecu_1=\Pi^{-1/2}\vpi$, and the remaining eigenvalues are all $-Z$.
Here, $\Pi=\diag(\vpi)$, $\vpi=[\pi_1,\dots,\pi_d]^\top$, $\pi_n=e^{-\beta\epsilon_n}/Z$, and $Z=\sum_{n=1}^de^{-\beta\epsilon_n}$.
Specifically, the spectral decomposition of $\msf{X}^{\rm ref}$ can be expressed as
\begin{equation}
	\msf{X}^{\rm ref}=\sum_{n=2}^d(-Z)\vecu_n\vecu_n^\top.
\end{equation}
Here, the vectors $\{\vecu_n\}$ form an orthonormal basis, i.e., $\vecu_m^\top\vecu_n=\delta_{mn}$.
Due to the spectrum degeneracy of the matrix $\msf{X}^{\rm ref}$, one can always construct a basis such that $\vecu_2=\Pi^{-1/2}\vecr_2/\sqrt{\vecr_2^\top\Pi^{-1}\vecr_2}$ because $\vecu_1^\top\Pi^{-1/2}\vecr_2=\vpi^\top\Pi^{-1}\vecr_2=\vb*{1}^\top\vecr_2=0$.
Using such a basis and following the idea in Ref.~\cite{Klich.2019.PRX}, we slightly modify $\msf{X}^{\rm ref}$ to obtain a new matrix $\msf{X}$ that has the same eigenvectors but different eigenvalues,
\begin{equation}
	\msf{X}=\msf{X}^{\rm ref}+\sum_{n=2}^d\Delta_n\vecu_n\vecu_n^\top=\sum_{n=2}^d(-Z+\Delta_n)\vecu_n\vecu_n^\top.
\end{equation}
Here, $Z>\Delta_2>\dots>\Delta_d\ge 0$ are small numbers that ensure the positivity of ${x}_{mn}$ for $m\neq n$.
Now, consider the matrix $\msf{W}\coloneqq\Pi^{1/2}\msf{X}\Pi^{-1/2}$.
Evidently, all off-diagonal elements of $\msf{W}$ are positive, and 
\begin{align}
	\vb*{1}^\top\msf{W}&=\sum_{n=2}^d(-Z+\Delta_n)\vb*{1}^\top\Pi^{1/2}\vecu_n\vecu_n^\top\Pi^{-1/2}\notag\\
	&=\sum_{n=2}^d(-Z+\Delta_n)\vecu_1^\top\vecu_n\vecu_n^\top\Pi^{-1/2}\notag\\
	&=\vb*{0}.
\end{align}
Therefore, $\msf{W}$ is the transition matrix.

We can show that this matrix satisfies all the conditions (1), (2), and (3).
First, its eigenvalues are distinct, $0=\lambda_1>\lambda_2>\dots>\lambda_d$, where $\lambda_n=-Z+\Delta_n$ for any $n\ge 2$.
Second, $\msf{W}$ fulfills the detailed balance condition since $\msf{X}=\Pi^{-1/2}\msf{W}\Pi^{1/2}$ is a symmetric matrix.
Lastly, the right eigenvector of $\msf{W}$ that corresponds to the second largest eigenvalue $\lambda_2$ is $\Pi^{1/2}\vecu_2=\vecr_2/\sqrt{\vecr_2^\top\Pi^{-1}\vecr_2}\propto\vecr_2$.
\end{proof}